\documentclass[12pt]{article}
\usepackage{amsmath}
\usepackage{graphicx}
\usepackage{enumerate}
\usepackage{natbib}
\usepackage{url} 

\newcommand{\blind}{1}

\usepackage[utf8]{inputenc} 
\usepackage[T1]{fontenc} 
\usepackage{hyperref} 
\usepackage{url} 
\usepackage{multirow}
\usepackage{booktabs} 
\usepackage{amsfonts} 
\usepackage{nicefrac} 
\usepackage{microtype} 
\usepackage{xcolor} 
\usepackage{multirow}
\usepackage{graphicx}

\usepackage{amsmath,amsthm,amssymb}
\usepackage{comment, times,setspace}
\usepackage{bm}
\usepackage{float}
\usepackage{enumitem}
\usepackage{url,hyperref} 
\usepackage{xr} 
\usepackage{tikz}

\usepackage{algorithm}
\usepackage{algpseudocode}

\usepackage{overarrows}
\usetikzlibrary{arrows.meta}
\usepackage{orcidlink}

\usepackage{stmaryrd}
\SetSymbolFont{stmry}{bold}{U}{stmry}{m}{n}

\usepackage{mathrsfs}

\newtheorem{theorem}{Theorem}
\newtheorem{remark}{Remark}
\newtheorem{lemma}{Lemma}

\newtheorem{prop}{Proposition}

\newtheorem{assumption}{Assumption}

\def\spacingset#1{\renewcommand{\baselinestretch}%
{#1}\small\normalsize} \spacingset{1}

\newcommand{\argmin}{\operatorname*{\arg\min}}

\newcommand{\E}{\operatorname{E}} 

\newcommand{\Var}{\operatorname{Var}}

\newcommand{\bb}{\mathbb}

\usepackage{stmaryrd}

\newcommand{\R}{\bb R}

\newcommand{\y}{\bm y}
\newcommand{\x}{\bm x}
\newcommand{\Z}{\bm Z}
\newcommand{\X}{\bm X}
\newcommand{\z}{\bm z}

\newcommand{\U}{\bm U}

\begin{document}

\def\spacingset#1{\renewcommand{\baselinestretch}%
{#1}\small\normalsize} \spacingset{1}

\date{}


\if1\blind
{
  \title{\bf{Conditional Data Synthesis Augmentation}
\thanks{This work was supported in part by NSF grant DMS-1952539 and NIH grants
R01AG069895, R01AG065636, R01AG074858, U01AG073079. (Corresponding author: Xiaotong Shen.) }}

\author{Xinyu Tian \thanks{Xinyu Tian is with the School of Statistics, University of Minnesota, MN, 55455 USA (email:
tianx@umn.edu)}, Xiaotong Shen \textsuperscript{\orcidlink{0000-0003-1300-1451}} \thanks{Xiaotong Shen is with the School of Statistics, University of Minnesota, MN, 55455 USA (email: xshen@umn.edu)}.}

  \maketitle
} \fi

\if0\blind
{
  \bigskip
  \bigskip
  \bigskip
  \begin{center}
    {\LARGE\bf Conditional Data Synthesis Augmentation}
\end{center}
  \medskip
} \fi

\bigskip
\begin{abstract}
Reliable machine learning and statistical analysis rely on diverse, well-distributed training data. However, real-world datasets are often limited in size and exhibit underrepresentation across key subpopulations, leading to biased predictions and reduced performance, particularly in supervised tasks such as classification. To address these challenges, we propose Conditional Data Synthesis Augmentation (CoDSA), a novel framework that leverages generative models, such as diffusion models, to synthesize high-fidelity data for improving model performance across multimodal domains, including tabular, textual, and image data. CoDSA generates synthetic samples that faithfully capture the conditional distributions of the original data, with a focus on under-sampled or high-interest regions. Through transfer learning, CoDSA fine-tunes pre-trained generative models to enhance the realism of synthetic data and increase sample density in sparse areas. This process preserves inter-modal relationships, mitigates data imbalance, improves domain adaptation, and boosts generalization. We also introduce a theoretical framework that quantifies the statistical accuracy improvements enabled by CoDSA as a function of synthetic sample volume and targeted region allocation, providing formal guarantees of its effectiveness. Extensive experiments demonstrate that CoDSA consistently outperforms non-adaptive augmentation strategies and state-of-the-art baselines in both supervised and unsupervised settings.
\end{abstract}

\noindent%
\vfill

\newpage
\spacingset{1.9} 
\section{Introduction}
Data augmentation is primarily utilized to improve the performance and robustness of machine learning models. This technique artificially generates diverse variations of training data through transformations while preserving the original labels. For instance, in image recognition tasks, one method involves flipping or rotating an image of a cat to create additional training samples, ensuring these variations are still accurately recognized as cats. When implemented effectively, data augmentation can significantly enhance model performance \citep{shorten2019survey}.

This technique has demonstrated strong empirical effectiveness in both image and text domains. For instance, translations, horizontal reflections, and intensity alterations were employed in \cite{krizhevsky2012imagenet} to augment data, thereby enhancing the accuracy of convolutional neural networks for image classification. Likewise, \cite{he2019bag} explored augmentation strategies such as Cutout, Mixup, and CutMix, which led to substantial performance gains. Specifically, the Cutout technique introduced by \cite{devries2017improved} involves randomly masking square regions of the input during training, thereby increasing robustness and accuracy. In natural language processing, \cite{xie2019unsupervised} applied back-translation to boost performance across various tasks. More recently, the focus has shifted to multimodal data augmentation—generating synthetic data that jointly models multiple modalities such as images, text, and audio \citep{baltruvsaitis2019multimodal, ramesh2021zero}. However, multimodal augmentation introduces unique challenges due to the inherent heterogeneity of data types, alignment issues, and complex dependencies that traditional methods do not address. An open question remains: can such data augmentation techniques further improve performance?

To address the challenges of data scarcity and representational imbalance in multimodal data, we propose a novel augmentation framework, Conditional Data Synthesis Augmentation (CoDSA). CoDSA leverages advanced generative modeling to expand training sample sizes while preserving the underlying data distribution. It operates by fine-tuning pre-trained generative models through transfer learning, enhancing the fidelity and relevance of synthetic data generation. A key feature of CoDSA is its use of conditional generation, guided by sample partitions defined by variables of interest, typically associated with underrepresented subpopulations. For example, in healthcare applications, such partitions may be based on demographic or clinical subgroups relevant to risk prediction.

By tailoring augmentation strategies to these targeted subgroups, CoDSA promotes model generalization in regions of the sample space where data are scarce or biased. Although CoDSA is compatible with a wide range of generative models, our implementation focuses on diffusion models due to their state-of-the-art performance in generating high-fidelity synthetic data \citep{dhariwal2021diffusion}. Our empirical results highlight the advantages of CoDSA over both non-augmentation baselines and non-adaptive approaches, the latter of which generate data unconditionally without prioritizing subpopulations of interest \citep{shen2023boosting}. In contrast, \cite{nakada2024synthetic} utilizes transformer-based language models for unconditional tabular data generation and subsequently extracts minority group samples to correct class imbalance and reduce spurious correlations in binary classification. Whereas both \cite{shen2023boosting} and \cite{nakada2024synthetic} generate data without direct conditioning, CoDSA directly employs conditional diffusion models \citep{ho2020denoising, kotelnikov2023tabddpm}, enabling adaptive, multimodal augmentation that is broadly applicable across supervised and unsupervised learning settings. This approach also provides a unified strategy for synthesizing structured and unstructured data, enhancing its utility in diverse domains.

The contributions of this paper are as follows:

1). \textbf{Adaptive synthetic data augmentation for multimodal data:} 
CoDSA employs a conditional generative model, such as a diffusion model, to generate multimodal synthetic samples tailored to user-defined regions of interest. These regions, specified by key variables (e.g., responses, features, or both), are selectively augmented to enhance various learning tasks. By focusing on areas requiring improvement, CoDSA dynamically adjusts the synthetic data volume to balance two critical trade-offs: (i) estimation and approximation errors related to data volume and (ii) the domain adaptation index alongside the weighted generation error associated with synthetic data allocation. Numerical examples in supervised and unsupervised tasks, supported by the statistical learning theory in Theorem~\ref{thm1}, demonstrate CoDSA’s superiority over traditional approaches. Consequently, when high-fidelity synthetic data is employed, CoDSA is expected to yield significant performance improvements compared to non-augmentation baselines.

2). \textbf{Optimal augmentation:} 
Unlike heuristic oversampling methods, such as SMOTE \citep{chawla2002smote} and its refinement ADASYN \citep{he2008adasyn} for imbalanced classification, or SMOGN for imbalanced regression \citep{Branco2017SMOGN}, CoDSA adaptively generates synthetic samples tailored to each local region of the feature space. It tunes the amount of augmentation via performance on a held-out validation set, recognizing that more synthetic data does not automatically translate into better results when generation errors accumulate. Extensive experiments demonstrate that CoDSA consistently exceeds both non-adaptive baselines and current state-of-the-art approaches across a range of multimodal tasks, including imbalanced classification, cross-domain sentiment analysis, facial age estimation, and image-based anomaly detection.

3). \textbf{Statistical efficiency and transfer learning:} 
CoDSA improves the fidelity of synthetic data relative to the original data by incorporating transfer learning into the generative modeling process. This aspect enhances model efficiency and overall learning accuracy. Our theoretical framework offers rigorous statistical insights into the key factors driving performance gains through data augmentation in multimodal settings, effectively bridging theory and practice. To our knowledge, this is the {\it first framework of its kind}. These findings confirm that high-fidelity synthetic data, facilitated through transfer learning, can significantly enhance model performance.

This article consists of five sections. Section \ref{sec:ssv} introduces the CoDSA methodology, emphasizing conditional synthetic data generation for multimodal datasets. Section \ref{sec:theory} develops a general theoretical framework highlighting CoDSA’s advantages over non-adaptive methods and presents illustrative examples of domain adaptation with latent diffusion models in imbalance classification. Section \ref{sec:numerical} provides numerical experiments demonstrating CoDSA’s performance relative to strong competitors
in the literature. Section \ref{discussion} discusses CoDSA’s broader implications and outlines directions for future research. Finally, the Supplementary Appendix includes technical proofs and experiment details.

\section{Multimodal data synthesis augmentation} \label{sec:ssv}

Consider a machine learning task where the training data $\Z$ is drawn from a distribution $P$, and $\Z$ may comprise multimodal data types such as tabular, image, and textual inputs. Model performance is evaluated on a separate distribution $Q$, which may differ from the training distribution $P$. The variable $\Z$ can represent the joint sample $(\X, Y)$ in supervised learning, or $\Z$ in unsupervised settings. This general formulation extends the classical machine learning paradigm, which assumes $P=Q$, to more realistic scenarios where this assumption no longer holds. In practice, such distributional shifts between $P$ and $Q$ are common and can significantly degrade model reliability and predictive accuracy.

 To illustrate the importance of addressing these shifts, we highlight two representative cases:

{\bf Imbalanced classification.} In binary classification tasks where $\Z=(\X, Y)$, with $\X$ denoting the predictor variables and $Y \in \{0, 1\}$ the class labels, class imbalance occurs when one class, typically the minority, is underrepresented in the training data. While the evaluation distribution $Q$ may assume balanced class proportions, the training distribution $P$ is often skewed toward the majority class. This mismatch can introduce bias during model training, resulting in poor generalization to the evaluation distribution.

{\bf Domain adaptation.} Domain adaptation refers to settings in which a model trained on a source domain must generalize to a target domain that exhibits different distributional characteristics. For example, a model trained on clean images may struggle when deployed on noisy or degraded images. In the supervised setting with $\Z = (\X, Y)$, this may take the form of a response shift, where the marginal distribution of the labels $Y$ differs between $P$ and $Q$. In contrast, the conditional distribution $P(\X \mid Y)$ remains unchanged. Conversely, under covariate shift, it is the marginal distribution of the features $\X$ that changes between domains, with the conditional distribution preserved: $P(Y \mid \X) = Q(Y \mid \X)$ \citep{sugiyama2007covariate}.

These examples motivate a central question: How can we design principled data augmentation methods to mitigate data scarcity and class imbalance, thereby enhancing model robustness under distributional shift?

To address this, we propose a generative augmentation framework that synthesizes high-fidelity data consistent with the original distribution. This approach enriches the training dataset by generating representative and diverse samples that account for underrepresented subpopulations or domain-specific variations. Unlike traditional augmentation techniques, often limited to simple transformations (e.g., rotations, flips, or token substitutions), our method leverages deep generative models to capture and replicate complex, multimodal data structures.

In what is to follow, we introduce a conditional generative framework designed to augment multimodal datasets adaptively. This method applies to various data modalities and task settings and offers a unified, statistically grounded approach to improving generalization.

\subsection{Conditional synthetic data generation}
\label{CoDSA}

Ideally, multimodal synthetic data can be generated using advanced generative models, such as diffusion 
models \citep{kotelnikov2023tabddpm} or normalizing flows \citep{dinh2016density,kingma2018glow}, trained on original datasets and further enriched through pretrained models. However, 
some subpopulations defined by certain variables of interest are undersampled, leading to data scarcity for specific groups. Addressing data scarcity through targeted strategies is crucial for reliable model performance.

To address data scarcity effectively, we partition the sample space $\Omega$ of $(\Z_i)_{i=1}^n$ into $K$ regions, $\Omega=\cup_{k=1}^K C_k$, each representing distinct subpopulations. Some regions $C_k$ may be particularly scarce, necessitating targeted data augmentation. For instance, in binary classification, one might define partitions $C_1 = \{Y=1\}$ and $C_2 = \{Y=0\}$ to represent minority and majority classes, respectively. Targeted synthetic data generation within these partitions directly addresses data scarcity, improves minority representation, reduces bias toward the majority class, and enhances overall generalization, as shown in our numerical examples in Section \ref{im-classification}.

Specifically, conditional synthetic data generation entails sampling synthetic data $\tilde \Z_{kj}$; $j = 1, \dots, m_k$, from the conditional distribution $P(\Z \mid \Z \in C_k)$. This approach addresses data scarcity by enriching regions with limited coverage. Models such as diffusion and normalizing flows are particularly well-suited for this purpose, and in Section \ref{transfer-diffusion}, we detail conditional diffusion generation. Moreover, we can enhance the accuracy of diffusion generation by employing transfer learning to fine-tune pretrained models derived from relevant studies.

\subsection{Multimodal data augmentation}

The preceding discussion motivates the proposed methodology: \textbf{Conditional Data Synthesis Augmentation (CoDSA)}. CoDSA is structured into three steps:

\begin{enumerate}
\item[] \textbf{Step 1: Sample splitting:} 
To prevent data reuse or data snooping \citep{wasserman2009high, wasserman2020universal}, we randomly partition the training sample of size $n$ into two subsamples with a split ratio $r\in [0,1]$. The first subsample, denoted by $\Z_{g}$ (of size $n_{g} =\lfloor r n\rfloor$ ), is designated for training the generative model. The second subsample, denoted by $\Z_{r}$ 
(of size $n_{r}=n- \lfloor r n\rfloor$), is reserved for subsequent data augmentation.

\item[] \textbf{Step 2: Training a generator:} 
Using $\Z_{g}$, we train or fine-tune a generator, such as a diffusion model, to replicate the underlying data distribution. Once successfully trained, the generative model produces high-fidelity synthetic samples to augment the original dataset.

\item[] \textbf{Step 3: Data augmentation:} 
We combine the reserved subset \( \Z_{r} \) with synthetic samples \( \tilde \Z \) (of size \( m \)), generated by the trained generative model, to form the augmented dataset $\Z_{c} = \Z_{r} \cup \tilde \Z$.

For estimation, the empirical loss on the augmented sample is defined as
\begin{equation}
\label{CoDSA1}
L_{m,n_r}(\theta,\Z_{c}) = \frac{1}{m+n_{r}} \sum_{k=1}^K\Big( \sum_{j=1}^{m_k} \ell\big(\theta, \tilde \Z_{kj}\big)
+\sum_{j=1}^{n_{r,k}} \ell\big(\theta, \Z_{kj}\big)\Big),
\end{equation}
where \(\ell\) is a loss function that measures the estimation error for the model \(\theta\) based on multimodal data (see Section~\ref{sec:numerical} for an example). Here, \(\Z_{kj}\) (resp.\ \(\tilde{\Z}_{kj}\)) denotes the \(j\)-th sample in region \(C_{k}\) from the reserved (resp.\ synthetic) data, and \(n_{r,k}\) (resp.\ \(m_k\)) is the sample size of \(C_k\) in the reserved (resp.\ synthetic) data.

 Minimizing $L_{m,n_r}$ with respect to $\theta$ over the parameter space $\Theta$ yields the estimator
$
\hat{\theta}_{\lambda} = \arg\min_{\theta \in \Theta} L_{m,n_r}(\theta,\Z_{c}).
$
Here, the tuning parameter is denoted by \( \lambda = (\bm{\alpha}, m, r) \), with \( \bm{\alpha} = (\alpha_1, \ldots, \alpha_K) \) serving as an allocation vector (where \( \alpha_k = m_k/m \)) and \( r \) representing the split ratio.
\end{enumerate}

Below is the algorithm for CoDSA.

\begin{algorithm}[H]
\caption{CoDSA: Conditional Data Synthesis Augmentation}
\label{alg:CoDSA-Mixed}
\begin{algorithmic}[1]
\Require Training data $\{\Z_i\}_{i=1}^n$, partition $(C_k)_{k=1}^K$, allocation $\bm{\alpha}=(\alpha_1,\dots,\alpha_K)$ with $\sum \alpha_k=1$, synthetic size $m$ (with $m_k=\alpha_k m$), split ratio $r\in(0,1]$, loss $\ell(\cdot,\cdot)$, parameter space $\Theta$.
\Ensure Estimated parameter $\hat{\theta}$ with tuning $\lambda=(\bm \alpha,m,r)$ 

\State \textbf{Sample split:} Randomly stratified partition $\Z$ into $\Z_{g}$ and $\Z_{r}$ with
a ratio of $r$.
\State \textbf{Generator training:} Train a generative model on $\Z_{g}$ to learn $P(\Z\mid \Z\in C_k)$.
\State \textbf{Synthetic data:} For each $k$, generate $\tilde \Z=\{\tilde \Z_{kj}\}_{j=1}^{m_k}$ using the trained model.
\State \textbf{Data mixing:} Form $\Z_{c} = \Z_{r} \cup \tilde \Z$.
\State \textbf{Estimation:} Compute 
$\hat{\theta}_{\lambda}=\arg\min_{\theta\in\Theta} L_{m,n_r}(\theta,\Z_{c})$.
\State \Return $\hat{\theta}_{\lambda}$.
\end{algorithmic}
\end{algorithm}

The tuning parameter \( \lambda \) is optimized on a validation set of the evaluation distribution \( Q \), yielding the minimizer \( \hat{\lambda} \). The final estimator is given by $\hat{\theta}_{\hat{\lambda}}$.
Theoretical insights into the optimal choice of \( \lambda \) are provided in Theorem~\ref{thm1} in Section~\ref{sec:theory}.

CoDSA builds upon the strengths of the generation model by integrating real and synthetic data, ensuring more reliable generalization, especially when the generative model is insufficiently trained or when relying solely on training data for estimation lacks robustness.

\section{Statistical learning theory} 
\label{sec:theory}

\subsection{Model accuracy of CoDSA}
\label{sec:g-theory}

We define the true parameter as $\theta_0 = \arg\min_{\theta \in \Theta} R_Q(\theta)$, with $R_Q(\theta) =
\mathrm{E}_Q \ell(\theta,\Z)$, where $\Theta$ is the parameter space, $\ell$ is the loss function,
and $\mathrm{E}_Q$ is the expectation under the distribution $Q$. Synthetic samples $(\tilde \Z_{j})_{j=1}^m$ are
assumed conditionally independent given the original data $(\Z_i)_{i=1}^{n}$.

\begin{assumption}[Partition validity]
\label{A-0}
The distributions \(P\) and \(Q\) share the same conditional distribution on each region \(C_k\); that is,
$
Q(\Z \mid \Z \in C_k) = P(\Z \mid \Z \in C_k), \quad \text{for } k = 1, \dots, K.$
\end{assumption}

Assumption~\ref{A-0} rules out distributional shift \emph{within} every region $C_k$. Consequently, replacing real observations by synthetic ones cannot bias the conditional risk inside each region, and the overall bias is forced to come from cross-region mass reallocation, which is later controlled via the generator error term $G_{\lambda}$.  The assumption is first used in the bias–variance decomposition, where expectations under $P$ are swapped for those under
$Q$ on $C_k$.

\begin{assumption}[Lipschitz continuity and boundedness]
\label{A-1}
There exists a constant $\beta>0$ such that
$
\sup_{\theta \in \Theta,\;\z\neq\z'} \frac{|\ell(\theta,\z) - \ell(\theta,\z')|}{d(\z,\z')} \le \beta$,
and a constant $U>0$ so that $\sup_{\theta\in\Theta,\;\z} |\ell(\theta,\z)| \le U$, 
where $d(\z, \z')$ is a suitable distance metric on $\z$. For example, $d$ may be the Euclidean distance for tabular $\z$, or an embedding-based distance \citep{zhang2018perceptual} for image data.
\end{assumption}

Lipschitzness controls the sensitivity of the loss to small input perturbations, while boundedness prevents heavy-tailed behavior.

\begin{assumption}[Variance control]
\label{A-2}
For some $c_v > 0$, the variance satisfies
$\Var_Q\big(\ell(\theta, \Z) - \ell(\theta_0, \Z)\big)
\le c_v\, \rho(\theta,\theta_0)$,
with excess risk $\rho(\theta,\theta_0)=R_Q(\theta)-R_Q(\theta_0)$.
\end{assumption}

For each region of interest $C_k$, let $P_k=P(\Z|C_k)$ and $\tilde{P}_k = P(\tilde \Z|C_k)$. The Wasserstein-1 distance between measures $P_k$ and $\tilde{P}_k$ is given by
$
W(\tilde{P}_k,P_k) = \sup_{\|f\|_{\mathrm{Lip}}\le1} \Big|\int f(\z)\,(\mathrm{d}\tilde{P}_k(\z)-\mathrm{d}P_k(\z))\Big|$, 
$k=1,\cdots,K$. For $f : \mathcal{Z} \to \mathbb{R}$, we define the expression $\|f\|_{\mathrm{Lip}}$ by 
$
\|f\|_{\mathrm{Lip}} = \sup \left\{ \frac{|f(\z) - f(\z')|}{d(\z, \z')} :\z, \z' \in \mathcal{Z},\, \z \ne \z' \right\}.
$

Assumption~\ref{A-2} is a local Bernstein (or \emph{margin}) condition linking the variance of the loss difference to its mean.  It sharpens the concentration bound  by enabling the use of Bernstein’s inequality instead of a sub‑Gaussian bound, as in \citep{shen1994convergence}, which results in the rate $\varepsilon_{\lambda}$ of Theorem~\ref{thm1} depending on $\rho(\theta,\theta_0)$ rather than merely on the variance.

\begin{assumption}[Error for a conditional generator]
\label{A-3}
For some constant $c_g>0$ and error rate $\tau_{k,n_{g}}>0$,
\[
\mathbb{P}\Big( W(P_k,\tilde{P}_k) \ge \tau_{k,n_{g}} \Big) \le \exp\Big(-c_g\, n_{g}\, \tau_{k,n_{g}}\Big), \quad k=1,\dots,K.
\]
where $n_{g}$ is the generator's training size, $c_g > 0$ is a constant, and $\tau_{k,n_{g}} > 0$ is the error rate.
\end{assumption}

Assumption~\ref{A-3} quantifies the fidelity of the conditional generator.  The exponential tail condition is standard in modern nonparametric density estimation and is satisfied by diffusion models with sufficient training data (see Supplementary 
Appendix~\ref{transfer-diffusion}), where Lemma~\ref{thm_diff_general} and Proposition~\ref{thm_ug_detail} give generation-error rates for general and transfer-learned models, respectively.

We now define key quantities. Let $\lambda=(\bm\alpha,m,r)$.

\begin{itemize}
\item \emph{Estimation error:} Denote by $\delta_{m+n_r}$ the estimation error rate satisfying an inequality of the form with $0<\gamma<1$,
\begin{equation}
\label{entropy}
\int_{\gamma\delta_{m+n_r}^2/8}^{4 c^{1/2}_{v} \delta_{m+n_r}}
H^{1/2}(u, \mathcal{F}) \, du \leq \frac{\gamma^{3/2}}{2^{10}} (m+n_r)^{1/2} 
\delta_{m+n_r}^2, 
\end{equation}
where $H(\cdot, \mathcal{F})$ is the $L_2$-bracketing metric entropy of the class $\mathcal{F} = \{L_{m,n_r}(\theta, \z_{c}) : \theta \in \Theta\}$, measuring the complexity of the function space $\mathcal{F}$; see \cite{tian2024enhancing} for details.

\item \emph{Approximation error:} Define $\rho(\theta_Q, \theta_0)$, where $\theta_Q = \arg\min_{\theta \in \Theta} 
R_Q(\theta)$ is the approximating point in $\Theta$ to $\theta_0$.

\item \emph{Domain adaptation index:} Define $D_{\lambda} = \sum_{k=1}^{K} \Big|\tilde{\alpha}_k-q_k\Big|$,
where $q_k=Q(\Z\in C_k)$, $\tilde{\alpha}_k = \frac{\alpha_k m + p_k n_{r}}{m+n_{r}}$, and $p_k=\frac{n_k}{n}$.
Note that $\tilde{\alpha}_k$ represents the proportion of the augmented (original and synthetic) samples allocated
to $C_k$. In the absence of domain adaptation (e.g., when $P = Q$), we have $D_{\lambda} = 0$ or equivalently $\tilde{\alpha}_k = q_k$; $k = 1, \dots, K$.

\item \emph{Generation error index:} Define $G_{\lambda} = \frac{m}{m+n_{r}} \sum_{k=1}^{K} \alpha_k\, \tau_{k,n_{g}}$, 
which represents a weighted aggregation of generation errors across the regions $C_k$'s.
\end{itemize}

Theorem \ref{thm1} below establishes a probabilistic error bound on the model performance 
for multimodal data, as defined by \eqref{CoDSA1}, using CoDSA with any generator, including diffusion models.

\begin{theorem}[CoDSA's accuracy]
\label{thm1}
Under Assumptions~\ref{A-0}--\ref{A-3}, for any
\begin{equation}
\label{rate}
\varepsilon_{\lambda} \ge \max\Big\{\rho(\theta_Q,\theta_0),\, \delta_{m+n_{r}},\, D_{\lambda},\, G_{\lambda}\Big\},
\end{equation}
we have 
\[
\mathbb{P}\Big(\rho(\hat{\theta}_{\lambda},\theta_0) \ge c_e\, \varepsilon_{\lambda}\Big)
\le 4\exp\Big(-\frac{1-\gamma}{c_k}\,(m+n_{r})\varepsilon_{\lambda}\Big)
+\sum_{k=1}^K \exp\Big(-c_g\, n_{g}\, \tau_{k,n_{g}}\Big),
\]
where $\mathbb{P}$ denotes the probability, $\gamma\in\Big(\frac{U}{4c_v},1\Big)$ is a constant, and
$c_e = 2 + 2U + 2\lambda$ and $c_d = 16\Big(c_v+c_v c_e + 4U^2 + 4U\beta\Big) + \frac{2U}{3}$. 
\end{theorem}

\begin{remark}
Theorem~\ref{thm1} establishes a probabilistic error bound for the CoDSA method. This result holds for both structured and unstructured $\Z$ in non-Euclidean spaces, provided that the associated generation error is appropriately controlled. Under the conditions $n_g \tau_{k,n_g} \to \infty$ and $(m + n_r) \varepsilon_\lambda \to \infty$ as $m, n \to \infty$, the excess risk 
$\rho(\hat{\theta}_\lambda, \theta_0)$ satisfies $\rho(\hat{\theta}_\lambda, \theta_0) = O_p(\varepsilon_\lambda)$, where $O_p(\cdot)$ denotes stochastic boundedness.
\end{remark}

Theorem \ref{thm1} indicates that CoDSA’s error rate, denoted by \(\varepsilon_{\lambda}\), is influenced by four key components: the approximation error \(\rho(\theta_Q, \theta_0)\), the estimation error \(\delta_{m+n_r}\) reflecting the complexity of the function class, the domain adaptation index \(D_{\lambda}\), and the generation error index \(G_{\lambda}\). We will discuss how these terms are impacted by the hyperparameter $\lambda$ in the next subsection.

\textbf{CoDSA versus the traditional method without synthetic data.}
In the traditional approach, optimization problem \eqref{opt} is solved by setting the tuning parameter to $\lambda = (\alpha, m, r) = (\tilde{\alpha}, 0, 0)$ with $\tilde{\alpha} = (p_1, \cdots, p_K)$. From this viewpoint, we may interpret the traditional method as a special case of 
CoDSA under a specific choice of tuning parameter. Consequently, when \(\lambda\) is adaptively optimized, CoDSA performs at least as well as the traditional method. Crucially, 
when the generation errors $(\tau_{k,n_{g}})_{k=1}^K$ are small, the excess error $\varepsilon_{\lambda}$ can converge at a faster rate with a large value of $m$ and a small $n_{r}$, 
yielding $\delta_{m+n_{r}} \ll \delta_n$ and and thereby improving performance through data augmentation. On the other hand, the generation errors are large,  choosing a smaller $m$ and a larger $n_{r}$ may be beneficial since the weighted aggregation of generation error  $G_{\lambda}$ increases in $m$ and decreases in $n_{r}$. We may improve by refining the allocation vector $\bm{\alpha}$ to reduce $G_{\lambda}$. See Section   \ref{t-example} for an illustrative example of imbalanced classification.

\textbf{CoDSA versus its nonadaptive counterpart.} The key difference between CoDSA and its nonadaptive variant \cite{shen2023boosting} lies in the data-dependent selection of the weight vector \(\bm{\alpha}\) and the use of adaptive sample splitting. 
The nonadaptive method corresponds to a fixed configuration \((\tilde{\alpha}, m, r)\), representing a restricted case of the broader CoDSA framework. Although the nonadaptive method may offer intermediate performance between the traditional approach and CoDSA, the adaptive weight selection in CoDSA plays a crucial role in minimizing the overall error. This adaptivity enables CoDSA to outperform its nonadaptive counterpart in practice.

\subsection{Performance of CoDSA in synthetic sample allocation}

To provide an insight into the performance of CoDSA’s estimator $\hat{\theta}_{\lambda}$ as a function of 
\(\lambda = (\bm{\alpha}, m, r)\), we investigate the ideal choice of $\lambda$. 

As suggested by \eqref{rate}, there appears a two-level trade-off:
\begin{itemize}
\item \textbf{First trade-off of $\Theta$:} There is a balance between the approximation error \(\rho(\theta_Q,\theta_0)\) and the estimation error, analogous to the classical bias-variance trade-off in supervised learning. The optimal trade-off is realized when these two error components are roughly equal, i.e., \(\rho(\theta_Q,\theta_0) = \delta_{m+n_{r}}\). This equality defines the best attainable rate \(\xi_{m+n_{r}} = \max\Bigl(\rho(\theta_Q,\theta_0),\ \delta_{m+n_{r}}\Bigr)\).
\item \textbf{Second trade-off of $\lambda$:} Given the first trade-off, another trade-off emerges among the estimation error $\xi_{m+n_{r}}$, the domain adaptation index \(D_{\lambda}\), and the generation error index \(G_{\lambda}\). This subsequent trade-off can be optimized in \(\lambda\), as outlined in \eqref{opt}.
\end{itemize}

Therefore, we seek the optimal \(\lambda\) in \eqref{rate}:
\begin{align}
\label{opt}
\min_{\lambda} \varepsilon_{\lambda} =
\min_{\lambda} \max\Bigl(\xi_{m+n_{r}},\, D_{\lambda},\, G_{\lambda}\Bigr),
\end{align}
where \(\xi_{m+n_{r}}\) (which converges to 0 as \(m+n_{r} \rightarrow \infty\)) is independent of \(\bm \alpha\) and thus $D_{\lambda}$. Ideally, one may directly optimize \eqref{opt} over \(\lambda\) to obtain the optimal tuning parameter.

However, it is possible to simplify the optimization process by considering a suboptimal solution. To optimize \eqref{opt}, note that \(D_{\lambda} \geq 0\). If we choose 
\begin{equation}
\label{eq-alpha}
\alpha^o_k = q_k + \frac{n}{m}(1-r)(q_k-p_k), \quad k=1,\cdots,K,
\end{equation}
then \(D_{(\bm \alpha^o, m, r)} = 0\) for any \(m\) and \(r\) provided that $
m \geq \max_k \frac{n(1-r)(q_k-p_k)}{I(q_k-p_k\geq 0)-q_k}$,
ensuring that \(\alpha^o_k \in [0,1]\). Consequently, one obtains
\begin{align}
\label{opt2}
\min_{\lambda} \varepsilon_{\lambda} \leq
\min_{(m,r)} \varepsilon_{(\bm \alpha^o, m, r)} =
\min_{(m,r)} \max\Bigl(\xi_{m+n_{r}},\, G_{(\bm \alpha^o, m, r)}\Bigr).
\end{align}

From \eqref{opt2}, we observe a trade-off between \(r\) and \(m\). The generation error \(\tau_{k, n_g}\) decreases as \(r\) increases, while the coefficient \(\frac{m}{m + n - nr}\) increases with \(r\), suggesting the existence of an optimal trade-off point \(r^o\). For $m$, the first term \(\xi_{m + n_r}\) typically decreases with \(m\), whereas \(G_{\lambda}\) increases with \(m\), indicating the presence of an optimal trade-off point \(m^o\).



\subsection{Theoretical example: Imbalanced classification}
\label{t-example}

This subsection demonstrates how one applies the theoretical framework in Theorem~\ref{thm1}  to assess the performance of the CoDSA method in the context of imbalanced binary classification. Specifically, we leverage a latent diffusion model for data augmentation to generate synthetic samples, thereby enhancing the representation of the minority class.

Consider a classification setting characterized by predictor variables $\X \in \mathcal{X} \subset \mathbb{R}^{d_x}$ (e.g., images or textual data) and binary labels $Y \in \{0,1\}$, with $Y=0$ indicating the minority class. Our objective is to estimate the conditional probability:
$\theta(\x) = P(Y = 1 \mid \X = \x)$, $\x \in \mathcal{X}$, from a chosen hypothesis space $\Theta$ using logistic loss:
$\ell(\theta(\x), y) = -y \log \theta(\x) - (1 - y) \log(1 - \theta(\x))$.
To address the class imbalance during evaluation, we adopt a balanced distribution over the labels such that $Q(Y=0) = Q(Y=1) = 1/2$. This adjustment mitigates bias toward the majority class, ensuring a fair assessment of model performance.

\paragraph{Latent diffusion models.}
The synthetic data generation involves a latent diffusion model composed of two essential components:

\begin{enumerate}
\item A pretrained \textbf{encoder-decoder pair $(f,g)$}:
\begin{itemize}
\item For image data, $f$ could be a the Variational Autoencoder (VAE) encoder \citep{kingma2013auto} that maps images to a lower-dimensional latent representation $\bm{u}$, and $g$ is a VAE decoder reconstructing images from latent variables. For example, 
in our age prediction example in Section \ref{age-prediction}, we utilize a stable diffusion model, which is itself
a latent diffusion model with both an encoder and a decoder.

\item For text data, $f$ is a transformer-based encoder (e.g., BERT) that converts sentences into latent embeddings, while $g$ is a transformer-based decoder that reconstructs sentences from these embeddings. In the sentiment analysis example (Section \ref{sentiment-analysis}), SentenceBERT \citep{reimers2019sentence} acts as the encoder, and a pretrained Vec2Text model \citep{morris2024language} decodes the embeddings to generate text reviews.

\end{itemize}

\item A diffusion model, trained to learn and replicate the distribution of the latent variables. The diffusion process involves adding noise gradually to the latent variables (forward diffusion) and training a neural network to remove this noise (backward diffusion), as defined in \eqref{forward} and \eqref{reverse} of Appendix \ref{transfer-diffusion}.
\end{enumerate}

Specifically, given a random sample $(\X_i, Y_i)_{i=1}^{n_{g}}$, the encoder produces latent representations $(\U_i)_{i=1}^{n_{g}}$, with $\U = f(\X_i)$. The diffusion model then learns conditional densities $p_{\bm{u}|C_k}$ for each class $C_k=\{Y=k\}$ in the latent space. Synthetic latent samples $(\tilde{\U}_i)_{i=1}^{m}$ generated by this process are subsequently transformed back into synthetic data $(\tilde{\X}_i)_{i=1}^{m}$ through the decoder: $\tilde{\X}_i=g(\tilde{\U}_i)$. Figure \ref{fig:da} presents the analysis procedure.

\begin{figure}[ht]
\centering
\includegraphics[width=0.45\textwidth]{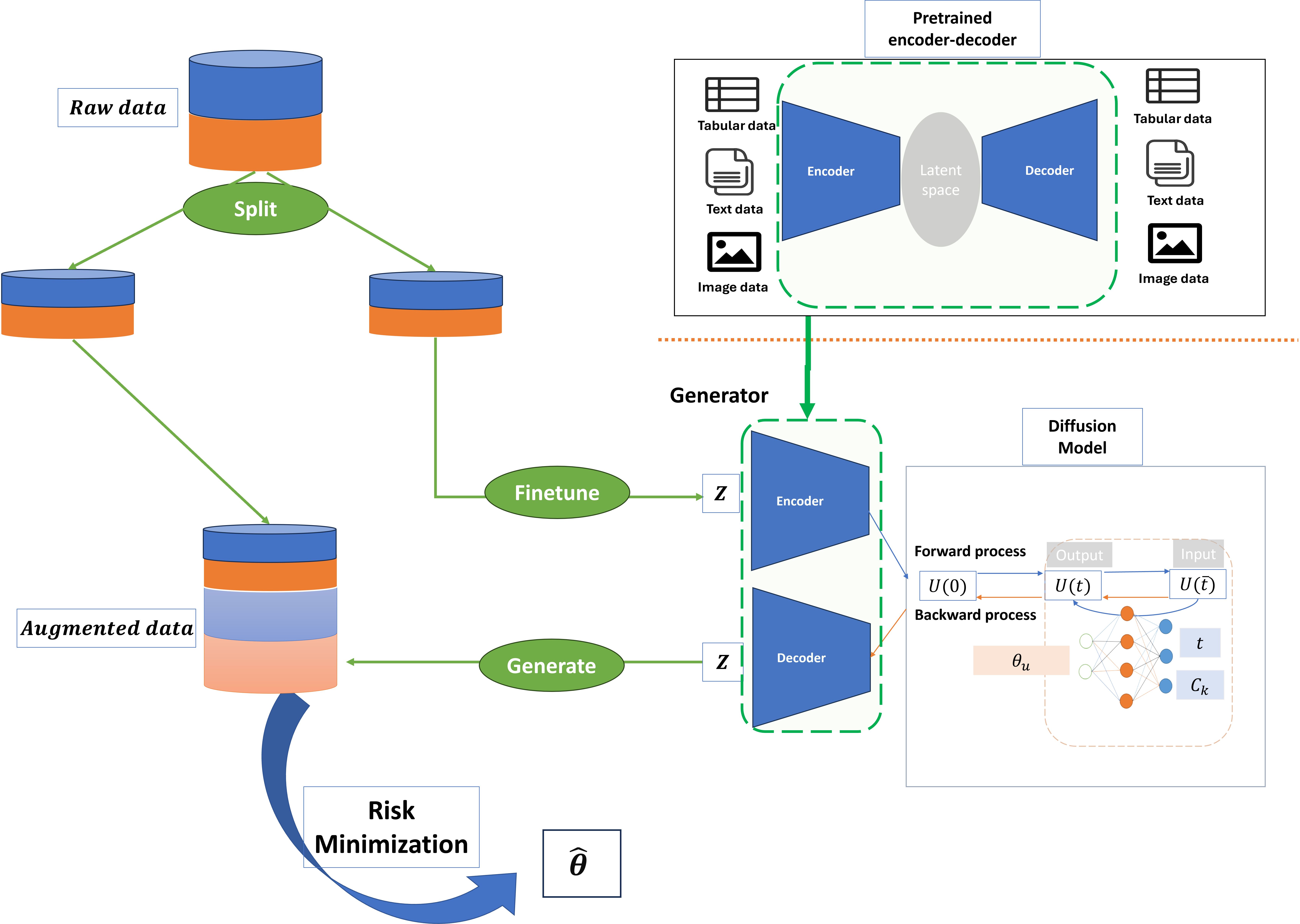} 
\caption{Flowchart for CoDSA via a latent diffusion model with an encoder-decoder architecture.}
\label{fig:da}
\end{figure}

Under appropriate smoothness conditions on the encoder-decoder architecture and the latent distribution, explicit bounds can be derived on the accuracy of synthetic data generation in the original multimodal space, measured by the Wasserstein distance $W(P_k,\tilde{P}_k)$ 
can be established. 

First, we define the smooth class before presenting the required conditions.

\paragraph{Smoothness.}
Consider a multi-index $\bm{\alpha}$ such that $|\bm{\alpha}| \leq \lfloor r \rfloor$, where $\lfloor r \rfloor$ denotes the integer part of smoothness $r>0$. A H\"older ball $\mathcal{C}^{r}(\mathcal{D}, \mathbb{R}^m, B)$ of radius $B$ with smoothness degree $r$ is defined as:
{
\small
\[
\left\{(g_1,\dots,g_m): \max_{1 \leq l \leq m}\left(\max_{|\bm{\alpha}| \leq \lfloor r \rfloor}\sup_{\x} |\partial^{\bm{\alpha}} g_l(\x)| + \max_{|\bm{\alpha}|=\lfloor r \rfloor}\sup_{\x \neq \y}\frac{|\partial^{\bm{\alpha}} g_l(\x) - \partial^{\bm{\alpha}} g_l(\y)|}{\|\x-\y\|^{r-\lfloor r \rfloor}}\right) < B\right\}.
\]
}

\begin{assumption}[Latent density]\label{A_ex}
Assume the pretrained encoder-decoder $(f,g)$ satisfies:
\begin{enumerate}
\item Reconstruction error bound:
$\E\|g(f(\X)) - \X\| \leq \varepsilon_s$,
with a small error rate $\varepsilon_s$ due to extensive pretraining.

\item Decoder boundedness: $\sup_{\bm{u}} \|g(\bm{u})\| \leq B_x$.

\item Conditional latent density regularity:
$p_{\bm{u}|y=k}(\bm{u}|y=k) = \exp\left(-\frac{c_k \|\bm{u}\|^2}{2}\right)f_k(\bm{u})$,
where $f_k \in \mathcal{C}^{r_u}(\mathbb{R}^{d_u}, \mathbb{R}, B_k)$ bounded below by a positive constant; $k=0,1$.
\end{enumerate}
\end{assumption}

Assume that the true conditional probability function $\theta_0$ is sufficiently 
smooth:
\begin{assumption}[Conditional probability]\label{A_f}
The true conditional probability $\theta_0$ belongs to a H\"older class:
$\theta_0 \in \mathcal{C}^{r_{\theta}}([-B_x,B_x]^{d_x}, [0,1], B)$,
and is uniformly bounded away from extremes:
$0 < \underline{c} \leq \theta_0(\x) \leq \overline{c} < 1$.
\end{assumption}

\paragraph{Neural networks.}
The conditional densities and probability functions are estimated by neural networks (NNs) constrained by architectural parameters (depth, width, sparsity) as specified by the parameter space \eqref{p-space} in Appendix \ref{transfer-diffusion}. ReLU activations are typically employed.

In the latent diffusion model, the ReLU neural networks estimate the score function in the diffusion process, which is described in detail in Appendix \ref{transfer-diffusion}. The network parameters are well set in Lemma \ref{thm_diff_general} of Appendix \ref{transfer-diffusion}.
In the classification model, we defined $\Theta$ as the ReLU neural network class with specific model settings defined in Lemma \ref{l_classification} of Appendix \ref{sec:proof}.

 Under these assumptions, we establish Theorem \ref{thm2}.

\begin{theorem}\label{thm2}
Under Assumptions \ref{A-0}--\ref{A_f}, the estimation error of CoDSA is upper bounded with $\varepsilon_{\bm{\alpha},m,r}=O\left(\max(\xi_{m+n_r}, G_{\lambda},  D_{\lambda})\right)$, where $\xi_{m+n_r} = O\left((m+n_{r})^{-\frac{r_{\theta}}{d_x+2r_{\theta}}}\log^{\beta_{\theta}}(m+n_{r})\right)$, $
    G_{\lambda}= O\left(\varepsilon_s + n_{g}^{-\frac{r_u}{d_u+2r_u}}\log^{\beta_u}(n_{g})\right)$,
    $D_{\lambda}=O\left(\big|\tilde{\alpha}_1 - \frac{1}{2}\big|\right)$,
and $\beta_u$ and $\beta_{\theta}$ are constants depending on the smoothness of latent and conditional probability functions, respectively.
\end{theorem}

The accuracy of the CoDSA method is characterized by the interplay among the generator training sample size $n_{g}$, the augmented sample size $(m+n_{r})$, and intrinsic smoothness parameters $(r_u, r_{\theta})$. Under suitable smoothness conditions on the encoder-decoder architecture and latent distribution, the estimation error $\varepsilon_s$ admits explicit convergence rates.
Specifically, assuming the encoder-decoder transformation is continuously differentiable and satisfies perfect compressibility as defined in \cite{jiao2024latent}, the estimation error $\varepsilon_s$, based on a pretraining dataset of size $N$, achieves the convergence rate:
$\varepsilon_s = O\left(N^{-\frac{1}{d_x + 2}}\right)$,
which becomes arbitrarily small as $N \to \infty$.

Furthermore, by selecting parameters $\tilde{\alpha}_1 = \frac{1}{2}$, $r=1$, and ensuring the augmented sample size $m$ is sufficiently large, the overall CoDSA error $\varepsilon_{\bm{\alpha},m,r}$ simplifies, up to higher-order terms, to:
$\varepsilon_{\bm{\alpha},m,r} \approx n^{-\frac{r_u}{d_u+2r_u}}\log^{\beta_u}(n) + m^{-\frac{r_{\theta}}{d_x+2r_{\theta}}}\log^{\beta_{\theta}}(m)$.
Moreover, the dimensionality reduction induced by the encoder-decoder typically leads to the condition:
$\frac{r_u}{d_u + 2r_u} > \frac{r_{\theta}}{d_x + 2r_{\theta}}$.
Consequently, by choosing sufficiently large augmentation size $m$, the error $\varepsilon_{\bm{\alpha},m,r}$ is dominated by the term:
$n^{-\frac{r_u}{d_u+2r_u}}\log^{\beta_u}(n)$,
which is strictly lower than the error rate of the traditional method without data augmentation, given by
$n^{-\frac{r_{\theta}}{d_x+2r_{\theta}}}\log^{\beta_{\theta}}n$.

\section{Numerical examples}
\label{sec:numerical}

In this section, we present numerical results on tabular, text, and image data to demonstrate the effectiveness of our CoDSA method. Detailed experimental settings are provided in Appendix~\ref{exp-details}. The implementation for the experiments presented in this section is available at: \url{https://github.com/shakayoyo/CoDSA}.

\subsection{Simulations: Imbalanced classification }
\label{im-classification}

This simulation study investigates the effectiveness of CoDSA in improving performance on imbalanced classification. In binary classification, we generate data by first sampling three latent variables $u_1-u_3$ from a mixture of normal distributions: $u_1 \sim 0.5\,N(-2,1) + 0.5\,N(2,1)$,
$u_2 \sim 0.5\,\text{Uniform}(0,1) + 0.5\,\text{Uniform}(2,3)$, and  $u_3 \sim \text{Exponential}(1) - 1$.
Using these latent variables, we construct ten nonlinear features:
\begin{align*}
\left\{
\begin{aligned}
x_1 &= u_1 u_2,      &\quad x_2 &= u_1 u_3,      &\quad x_3 &= u_2 u_3,  &\quad x_4 &= u_1^2,     &\quad x_5 &= u_2^2,  \\x_6 &= u_3^2, &\quad x_7 &= u_1 u_2 u_3,  &\quad x_8 &= u_1^3,        &\quad x_9 &= u_2^3,  &\quad x_{10} &= u_3^3
\end{aligned}
\right.
\end{align*}
yielding a 10-dimensional feature vector \(\x = (x_1, \dots, x_{10})\).

Next, to introduce a decision function, consider  
\[
s(\x) = \left( \sin\Bigl(2\pi\, \frac{\x_{1:5}^T \bm{w}}{1 + \left|\x_{1:5}^T \bm{w}\right|}\Bigr) \right)^2 - \left( \cos\Bigl(3\pi\, \frac{\x_{6:10}^T \bm{w}}{1 + \left|\x_{6:10}^T \bm{w}\right|}\Bigr) \right)^2,
\]
where the 5-dimensional weight vector \(\bm{w}\) is linearly spaced between $-1$ and $1$, 
\(\x_{1:5}\) denotes the subvector containing the first five components of \(\x\), and $^T$ denotes the transport.

 Let thresholds based on the quantiles of \(s(\x)\) be 
$\tau_{\text{lo}} = q\Bigl(\frac{n_1}{2(n_1+n_2)}\Bigr)$ and $\tau_{\text{hi}} = q\Bigl(1 - \frac{n_1}{2(n_1+n_2)}\Bigr)$,
with \(q(\cdot)\) the quantile function of \(s(\x)\). Now define a central threshold and a half-range as follows:
$\tau = \frac{\tau_{\text{lo}} + \tau_{\text{hi}}}{2}$ and $\delta = \frac{\tau_{\text{hi}} - \tau_{\text{lo}}}{2}$.
Then, the decision function that outputs a binary label is given by
$Y = \frac{1}{2}\left( \operatorname{sign}\Bigl(\delta - \bigl| s(\x) - \tau \bigr|\Bigr) + 1 \right)$.
Here, \(Y=1\) if \(|s(\x)-\tau| < \delta\) and \(Y=0\) otherwise, and $\operatorname{sign}$ denotes the sign function.

Applying this sampling scheme, we obtain a sample of size \( n = n_1 + n_2 = 5200 \), with
\( n_1 = 1400 \) and \( n_2 = 3800 \) samples belonging to class \( C_1 = \{Y = 0\} \) and 
class \( C_2 = \{Y = 1\} \), respectively. For evaluation, we adopt a balanced metric with equal class weights, \( Q(Y=0) = Q(Y=1) = \frac{1}{2} \), in contrast to the empirical class proportions \( P(Y=0) = \frac{n_1}{n} \) and \( P(Y=1) = \frac{n_2}{n} \). Accordingly, we construct balanced validation and test sets, consisting of 400 and 800 samples, respectively, with an equal number of observations from each class, while the remaining data is used for training. The validation set supports hyperparameter tuning and early stopping of the neural network, while the test set is for final performance evaluation. As a result, the evaluation distribution \( Q \) differs from the training distribution \( P \).

For classification, we adopt a three-layer feed-forward neural network with ReLU activations in the first two layers and 128 hidden units as our baseline. We benchmark two oversampling methods for imbalanced classification within the non-transfer and transfer variants of CoDSA: SMOTE (Synthetic Minority Over-sampling Technique) \citep{chawla2002smote} and ADASYN (Adaptive Synthetic Sampling) \citep{he2008adasyn}. In the non-transfer setting, we train a latent diffusion model on a held-out subset of the training data, and combine the generated synthetic samples with the remaining data before training the classifier. In the transfer setting, the latent diffusion model is guided by a pretrained encoder–decoder trained on an independent dataset of 10,000 examples that exhibit the same structural transformations, enabling cross-domain knowledge transfer.

We use cross-entropy, defined by the logistic loss $\ell$ for classification, as the evaluation metric on the test set. For a data point $\z = (y, \x)$, where $\x$ is the predictor and $y \in \{0,1\}$ is the label, the loss in \eqref{CoDSA1} is defined as:
$
\ell(\theta, z) = -\left[ y \cdot \log P(Y = 1 \mid \x) + (1 - y) \cdot \log (1 - P(Y = 1 \mid \x)) \right]
$,
where $P(Y = 1 \mid \x)$ is the model's predicted probability that the label is 1 given the input $\x$.

For CoDSA, optimal hyperparameters $\lambda = (\bm{\alpha}, m, r)$ are chosen via grid search to minimize validation cross-entropy.
See Appendix \ref{exp-details} for more details.

As shown in Table \ref{tab:class_results}, the baseline neural network attains the lowest cross-entropy on the majority class (0.1920 $\pm$ 0.0199). Still, it performs very poorly on the minority class (0.9594 $\pm$ 0.0688), giving the highest overall cross-entropy (0.5757 $\pm$ 0.0264). This pattern highlights the baseline model’s bias toward the majority class under severe class imbalance.

All four methods, SMOTE, ADASYN, non-transfer CoDSA, and transfer CoDSA, substantially reduce cross-entropy for both classes, confirming that synthetic augmentation effectively mitigates class imbalance. Among the traditional oversamplers, SMOTE (0.4739 $\pm$ 0.0365 overall) edges out ADASYN (0.4764 $\pm$ 0.0343), indicating that uniform interpolation slightly outperforms adaptive interpolation in this setting. Both, however, are surpassed by the CoDSA variants: non-transfer CoDSA reduces overall cross-entropy to 0.4641 $\pm$ 0.0167, and transfer CoDSA achieves the best result of 0.4589 $\pm$ 0.0153. The additional gain from the transfer CoDSA shows that leveraging knowledge from related tasks further improves the quality of the generated samples.

In the transfer CoDSA setting, the optimal hyperparameters were found to be
$(m/n = 1.44 \pm 0.59, \alpha_1 = 0.59 \pm 0.12, r = 0.67 \pm 0.24\bigr)$.
The tuning favored a large synthetic sample size, suggesting that the generative model produces high-fidelity data which effectively complements the original training set.  The fact that \(\alpha_1>0.5\) implies deliberate oversampling of the minority class, steering the augmented dataset toward class balance and reinforcing our theoretical result that near-balanced augmentation enhances generalization. Moreover, the observed variability in the optimal hyperparameters further supports the trade-off characterized by Theorem~1 during the tuning process.

\begin{table}[ht]
\centering
\caption{Cross‐entropy values (standard errors in parentheses) on the test set for the baseline, SMOTE, ADASYN, and non‐transfer and transfer CoDSA methods, over 10 resamples of the original data.}
\begin{tabular}{llccc}
\toprule
\textbf{Model}    &        & \textbf{Majority}          & \textbf{Minority}          & \textbf{Overall}             \\[3pt]
\midrule
Baseline &  (Neural Network)             & 0.1920 (0.0199)    & 0.9594 (0.0688)     & 0.5757 (0.0264)     \\[2pt]
SMOTE    &               & 0.4134 (0.0331)    & 0.5345 (0.0607)     & 0.4739 (0.0365)     \\[2pt]
ADASYN   &               & 0.4388 (0.0493)    & 0.5141 (0.0981)     & 0.4764 (0.0343)     \\[2pt]
CoDSA    & Non-Transfer  & 0.4441 (0.0625)    & 0.4841 (0.0523)     & 0.4641 (0.0167)     \\[2pt]
CoDSA    & Transfer      & 0.4224 (0.0505)    & 0.4955 (0.0520)     & 0.4589 (0.0153)    \\[2pt]
\bottomrule
\end{tabular}
\label{tab:class_results}
\end{table}

\begin{figure}[htbp]
\centering
\includegraphics[width=0.75\textwidth]{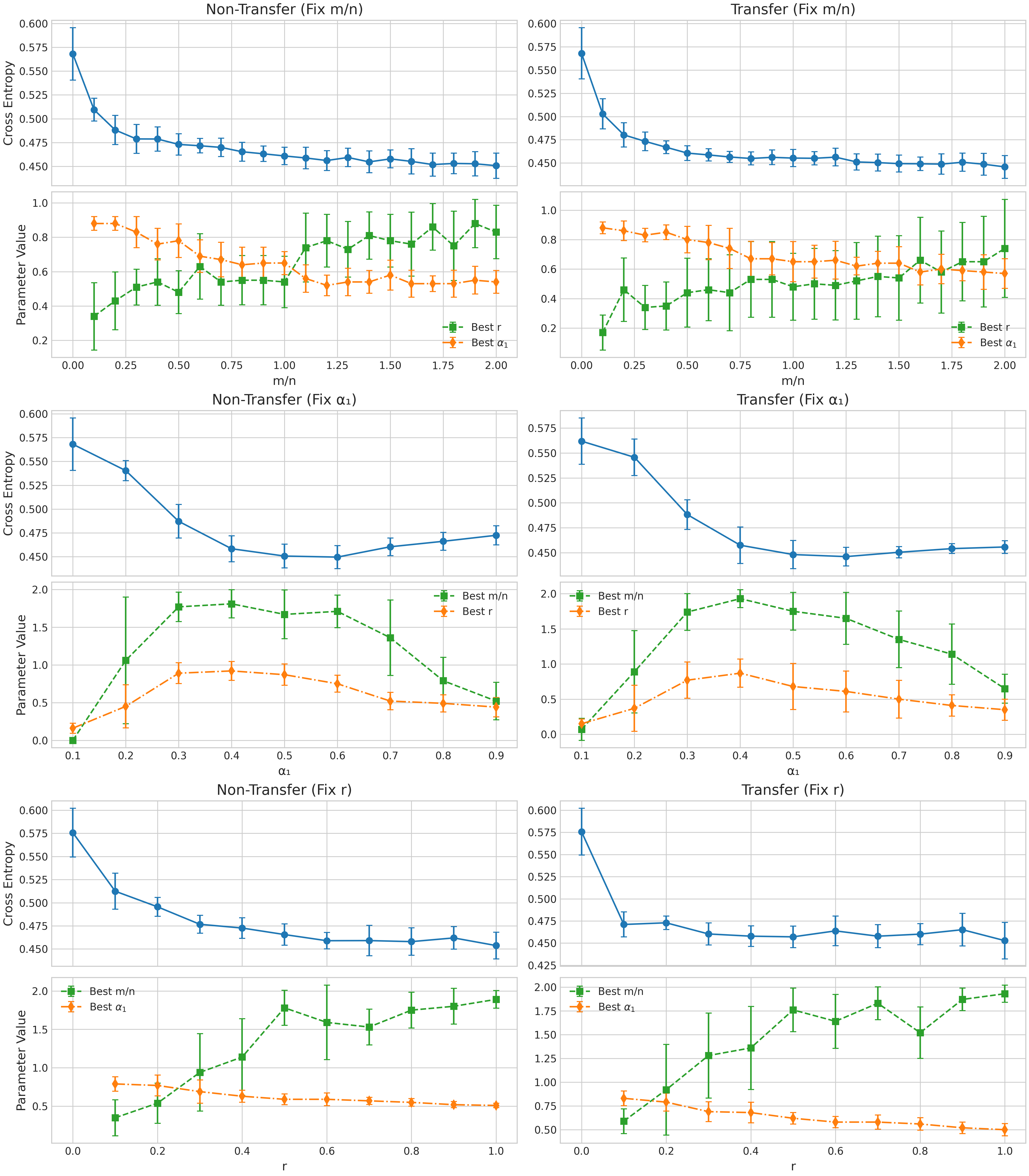} 
\caption{Test cross-entropy of CoDSA (mean$\pm$ one SE) across the three hyperparameters $m/n$, $\alpha_1$, and $r$ under the experimental conditions of Section~\ref{im-classification}. {\bf Left column}: non-transfer CoDSA; {\bf right column}: transfer CoDSA. {\bf Top row}: dependence on the synthetic-to-original sample ratio $m/n$; {\bf middle row}: dependence on the minority-region allocation proportion $\alpha_1$; {\bf bottom row}: dependence on the sample-split ratio $r$. Each point is the average of 10 simulation replicates (error bars denote  $\pm$ one standard error). }
\label{fig:toy}
\end{figure}

To elucidate how performance varies with each hyperparameter and how these hyperparameters interact, we present the marginal relationships in Figure~\ref{fig:toy}. Each point in the figure is obtained by fixing one hyperparameter and selecting the optimal values of the other two on validation error.

As shown in Figure~\ref{fig:toy}, increasing the synthetic-to-original ratio ($m/n$) sharply reduces test cross-entropy for both non-transfer and transfer CoDSA variants, although gains diminish for large values of $m/n$. Varying the allocation weight $\alpha_1$ uncovers a pronounced optimum: the lowest cross-entropy occurs when approximately 60\% of generated samples target the minority region in both settings. Meanwhile, increasing the split ratio $r$ above 0.6 yields superior performance, balancing the needs of generator training and downstream learning. The lower panels of each subplot illustrate hyperparameter interactions: as $m/n$ grows, the optimal split ratio rises while $\alpha_1$ converges toward 0.5, indicating that larger synthetic datasets require more data for high-fidelity generation and favor near-balanced augmentation to maximize performance.

\subsection{Simulations: Regression with undersampled regions}

This simulation study investigates how adaptive data augmentation can address data imbalance in regression settings, particularly in regions of the input space with sparse observations. As in imbalanced classification, where minority classes receive fewer samples, regression problems may exhibit ``undersampled regions'', subpopulations for which data is limited, leading to degraded model performance in those areas.

Consider a regime-switching regression model with region-dependent nonlinearity. Our primary objective is to demonstrate the effectiveness of adaptive data augmentation in improving prediction accuracy for a continuous outcome variable $y$, given a vector of predictors $\x$ with a focus on enhancing performance in the undersampled region.

 Specifically, we define two regions, $C_1$ (undersampling): $u_1 \sim \mathrm{Uniform}(0,0.5)$ and
$C_2$ (oversampling): $u_1 \sim \mathrm{Uniform}(0.5,1)$, corresponding to undersampled and oversampled areas, respectively.
Let $u_2 = u_1^2 + \varepsilon$; $\varepsilon \sim N(0,1)$. Then, we construct five features 
$\x = (x_1, \cdots, x_5)$:
$x_1 = u_1$; $x_2 = I(u_1 > 0.5)u_1 + u_2$; $x_3 = I(u_1 > 0.5)\,u_2 + \log(1 + |u_2|)$, $x_4 = |u_2|$;
$x_5 = u_1 - u_2$, where $I(\cdot)$ is the indicator function. 

Next, we set the regression parameter vector $\bm \beta = (3.0, 2.0, -1.0, 0.5, 1.0)$. Then, we introduce region-dependent nonlinearity into the 
regression function:
\[
f(x) =
\begin{cases}
2(\x \cdot \boldsymbol{\beta})^2 + \x \cdot \boldsymbol{\beta}, & \text{if } \x \in C_1,\\[1mm]
2(\x \cdot \boldsymbol{\beta})^2 - \x \cdot \boldsymbol{\beta}, & \text{if } \x \in C_2.
\end{cases}
\]
Then, the outcome $y$ becomes: 
$y = f(\x) + \varepsilon$, with $\varepsilon \sim N(0,\sigma^2)$,
where the choice of $\sigma=0.2$ controls the noise level in the observations. 

The combination of these two data types results in a random sample of size $n_1 = 1400$ from $C_1$ and a sample of 
size $n_2 = 3800$ from $C_2$, yielding a total of $n_1 + n_2 = 5200$ observations. From this sample, we extract a balanced validation set of 400 samples and a balanced test set of 800 samples while using the remaining observations for training.

To predict the outcome from predictors, we use a RandomForest regressor from \texttt{scikit-learn} \citep{breiman2001random}, employing the squared-error ($L_2$) loss for $\ell$ in~\eqref{CoDSA1} as our baseline model. We compare this baseline and a single oversampling strategy, SMOGN, against two variants of CoDSA. SMOGN \citep{Branco2017SMOGN} rebalances the training set by adding Gaussian-perturbed replicas to sparsely populated regions of the continuous target distribution. Non-transfer CoDSA trains a latent-diffusion model on a held-out subset of the original data and augments the remaining observations with the generated samples before fitting the Random Forest. Transfer CoDSA follows the same procedure but guides the latent-diffusion model with a pretrained encoder–decoder trained on an independent dataset of 10,000 examples that exhibit the same structural transformations, thereby enabling cross-domain knowledge transfer.

For non-transfer and transfer CoDSA models, we tune hyperparameters via grid search on the validation dataset. The tuning hyperparameters include the sample split ratio $r \in [0,1]$, region allocation ratio $\alpha_1 \in [0.1, 0.9]$, and synthetic-to-original sample size ratio $m/n \in [0,2]$. We use the validation set for tuning diffusion model parameters. Optimal hyperparameters are selected based on the root mean squared error (RMSE) on the validation dataset.

For evaluation, we report the RMSE on the test set. As shown in Table~\ref{tab:toy2}, both CoDSA variants outperform the baseline and the SMOGN.
Transfer CoDSA yields the largest gain, reducing the overall RMSE (1.0241~$\rightarrow$~0.5911), an improvement of roughly 42\% relative to the baseline and 30\% relative to SMOGN. Non-transfer CoDSA surpasses SMOGN, trimming the overall RMSE by about 4\% (0.8414~$\rightarrow$~0.8092).

These gains are driven primarily by substantial improvements on the minority class: transfer CoDSA cuts minority-class RMSE by about 60\% versus the baseline and 51\% versus SMOGN (1.5534~$\rightarrow$~0.7620), while non-transfer CoDSA achieves reductions of 43\% and 30\%, respectively (1.5534~$\rightarrow$~1.0915). Although all augmentation methods slightly increase the majority-class error, the net effect remains strongly positive for both CoDSA variants, especially in the transfer setting.

The hyperparameter search reveals that, in the transfer CoDSA setting, the model favors a minority-focused oversampling strategy—allocating over 80\% of synthetic samples to the minority class (\(\alpha_1 = 0.83 \pm 0.09\))—along with a split ratio of \(r = 0.43 \pm 0.21\) and a synthetic-to-original ratio of \(m/n = 1.16 \pm 0.54\). When combined with the pretrained encoder–decoder, this configuration delivers the greatest overall benefit.

\begin{table}[htbp]
\label{tab:toy2}
  \centering
\caption{Average test-set RMSE (standard error) for the Baseline, SMOGN, and CoDSA models across 10 simulation replications (on
a test set of size $800$; hyperparameters were tuned on a separate validation set of size $400$.}
  \label{tab:toy}
  \begin{tabular}{llccc}
    \toprule
    \textbf{Model} & & \textbf{Majority} & \textbf{Minority} & \textbf{Overall}\\[3pt]
    \midrule
    Baseline & (Random Forest) & 0.1392 (0.0224) & 1.9090 (1.2483) & 1.0241 (0.6262)\\
    SMOGN & & 0.2312 (0.1076) & 1.5534 (1.0576) & 0.8414 (0.4715)\\
    CoDSA & Non-Transfer & 0.5269 (0.2305) & 1.0915 (0.7772) & 0.8092 (0.3835)\\
    CoDSA & Transfer & 0.4201 (0.1046) & 0.7620 (0.5011) & 0.5911 (0.2718)\\
    \bottomrule
  \end{tabular}
\end{table}

\subsection{Sentiment analysis}
\label{sentiment-analysis}

This subsection evaluates the predictive performance of CoDSA against a baseline model that does not employ data augmentation or transfer learning to incorporate insights from related domains. The evaluation uses the Yelp Reviews dataset (\url{www.yelp.com/dataset}), a widely recognized benchmark for sentiment analysis and star rating prediction. The dataset comprises over 6.9 million reviews, each rated on a scale of 1 to 5 stars, with 5 indicating the highest level of satisfaction. In addition to review text, it includes business-related attributes such as categories, locations, and user feedback metrics (e.g., useful, funny, and cool votes).

From the first 100,000 short reviews ($<$ 64 words), we study domain adaptation for rating prediction, transferring knowledge from reviews deemed non-useful to those considered useful. Partitioning by helpful-vote count yields a source domain of 71,687 reviews with zero votes and a target domain of 28,313 reviews with at least one.

As shown in Figure~\ref{fig:hist_yelp}, the rating distributions for the source (non-useful reviews) and target (useful reviews) datasets differ significantly, with both exhibiting a strong skew toward higher ratings. However, insights from reviewers who mark reviews as non-useful can still provide valuable context for understanding the perspectives of those who find them useful. Notably, ratings are predominantly clustered around 4 and 5, resulting in an undersampling of ratings between 1 and 3.
For sentiment analysis, we utilize a multi-class logistic regression model with balanced class weights as the classifier, leveraging a TF-IDF feature extractor \citep{zhang2011comparative} to convert text reviews into numerical features. We then apply the CoDSA method for data augmentation. Let $\Z=(Y,\X)$ with $Y$ and $\X$ denoting ratings and text reviews, and let $C_k = {Y=k}$; $k = 1, \dots, 5$. Specifically, we employ a latent diffusion model with an encoder-decoder as our generator. SentenceBERT (Sentence-Bidirectional Encoder Representations from Transformers) \citep{reimers2019sentence} serves as the encoder, mapping sentence embeddings from $\Z$ to $\U$. A pretrained Vec2Text model \citep{morris2024language} then decodes the diffused vector embeddings to generate text reviews.

For comparison, we employ two training strategies: non-transfer CoDSA and transfer CoDSA, both utilizing latent diffusion models. Non-transfer CoDSA trains the latent diffusion model solely on the target dataset. In contrast, transfer CoDSA fine-tunes a diffusion model—originally pretrained on the source (non-useful reviews) data—using the target (useful reviews) dataset.

\begin{figure}[ht]
\centering
\begin{minipage}[t]{0.45\textwidth}
    \centering
    \includegraphics[width=\linewidth]{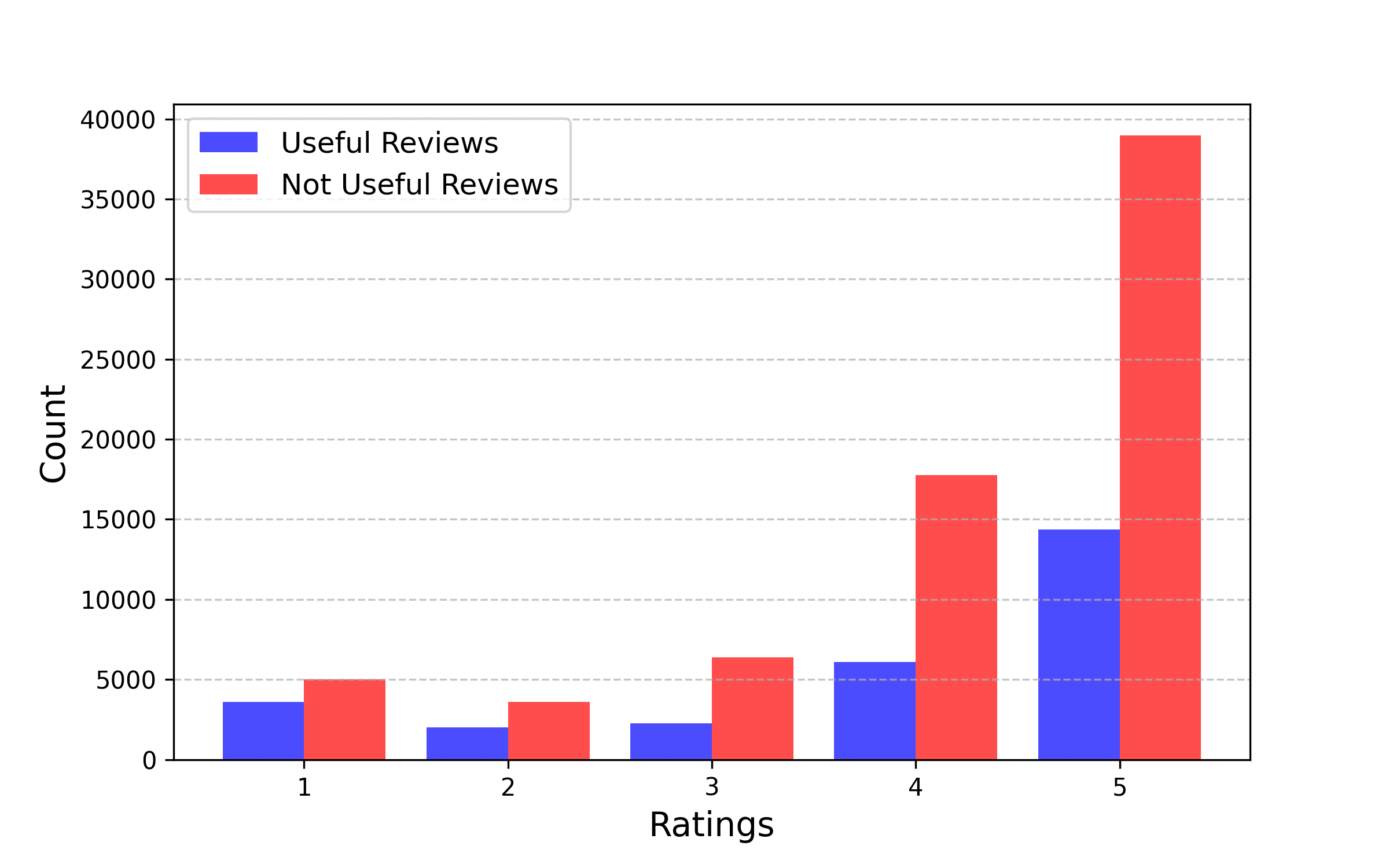}
    \caption{Distributions of ratings for the target (useful reviews) and source (non-useful reviews) data.}
    \label{fig:hist_yelp}
\end{minipage}
\hfill
\begin{minipage}[t]{0.45\textwidth}
    \centering
    \includegraphics[width=\linewidth]{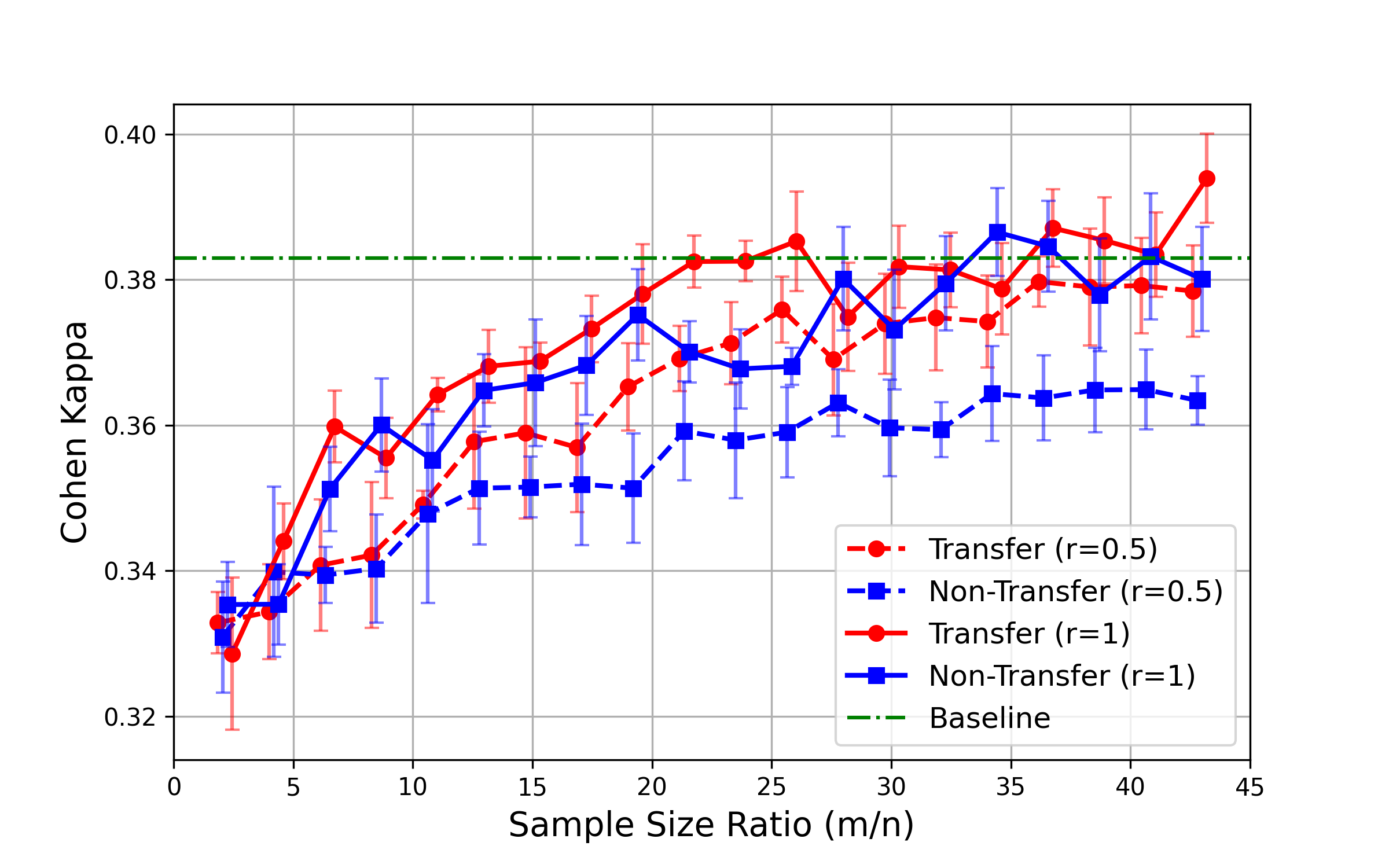}
\caption{Cohen’s $\kappa$ with standard error bar on the CoDSA test set, evaluated over five train–test splits and synthetic-to-training sample size ratios from $m/n = 0$ to $45$.}
    \label{fig:ga_yelp}
\end{minipage}
\end{figure}

Given trained diffusion models, we generate synthetic samples in a balanced manner, with $m_k=\alpha_k m$ and $\alpha_k=1/5$; $k=1,\cdots,5$. For non-transfer-CoDSA and transfer-CoDSA, we tune the synthetic sample size $m$ and the split ratio $r$ in the range $\{0,n,\cdots, 45 n\}$ and $\{0.5,1\}$ to reduce the computational cost. Given augmented samples, we then refit the logistic regression under the same settings as the baseline model to predict ratings from reviews in the test dataset. In other words, we apply the logistic loss to \eqref{CoDSA1} to obtain regression estimates.

To evaluate performance on ordered ratings ranging from 1 to 5, we adopt Cohen's Kappa statistic \cite{cohen1960kappa} as the primary evaluation metric. Cohen's Kappa ($\kappa$) is defined as:
$\kappa = \frac{p_o - p_e}{1 - p_e}$, where $p_o$ is the observed agreement between predictions and actual ratings, and $p_e$ represents the expected agreement by chance. Unlike the commonly used Root Mean Square Error (RMSE), Cohen's Kappa adjusts for chance agreement, making it particularly effective for evaluating performance on imbalanced categorical data.

Figure \ref{fig:ga_yelp} shows that augmenting the training data with synthetics lifts both transfer-CoDSA and non-transfer CoDSA above the baseline trained on the original $n=23{,}313$ real samples (green dashed line at $\kappa\!\approx\!0.383$). Allocating all real data to the generator ($r=1$) is superior: the transfer model peaks at roughly $m/n\!\approx\!43$ ( $m\!\approx\!1{,}000{,}000$) with $\kappa\!\approx\!0.395$, improving on the baseline by about $0.0120$
($+3.1\%$); the non-transfer model tops out near $m/n\!\approx\!34$ ($m\!\approx\!800{,}000$) at $\kappa\!\approx\!0.387$, a gain of roughly 0.0045 ($1.2\%$). In contrast, the partial-allocation setting ($r=0.5$) never reaches these levels, confirming that full allocation is preferable. Overall, transfer learning delivers about a modest improvement 
over the non-transfer variant, indicating that pre-training enables the model to glean richer domain-specific information from synthetic data that the non-transfer approach cannot fully exploit.

In summary, while increasing the volume of synthetic data improves non-transfer and transfer models, the gains are significantly higher for the transfer model. This result suggests that synthetic data alone may not necessarily enhance a model’s predictive power. However, when combined with domain adaptation via transfer learning, we can generate synthetic samples that exhibit meaningful variation and capture additional domain insights, leading to enhanced performance for rating prediction. In particular, incorporating insights from reviewers who find reviews non-useful helps improve rating prediction for those who find them useful.

\subsection{Age prediction from facial images}
\label{age-prediction}

This experiment utilizes the UTKFace dataset (\url{https://susanqq.github.io/UTKFace/}) for age prediction, a large-scale image dataset widely used in computer vision tasks such as age estimation, gender classification, and ethnicity recognition. The dataset contains over 23,705 images spanning a broad age range (0--116 years), multiple ethnicities, and genders. However, it exhibits a significant class imbalance, with 9,424 White, 4,578 Black, 3,584 Asian, 4,047 Indian, and 2,072 individuals from other ethnic backgrounds. 

We group the data into five ethnicity groups $C_k$ for $k=1,\dots,5$, defined by the ethnicity variable $E$. First, we randomly sample a test set and an independent validation set of 3,000 and 2,500 images, respectively, ensuring equal representation (600 and 500 images per ethnicity) for unbiased evaluation. Then, we use the remaining 18,205 images for training, with the validation set guiding early stopping and the test set for testing.

For the age prediction task, let $\Z$ denote the outcome variable $ Y$ (age), $X$ the input image, and $C_k$ the categorical ethnicity variable, $ k=1, \cdots, 5$. To address sampling bias, including under- and over-representation across ethnic groups, we apply the CoDSA method to augment the dataset and mitigate imbalance in \eqref{CoDSA1}:
$
\min_{\theta\in\Theta} \; \ell(\theta(\x),y)
$
with $\ell$ the $L_2$ loss, that is 
$\ell(\theta(\x),y)=(y-\theta(\x))^2$, where $\Theta$ is a class of convolutional neural network (CNN) models. 
We train the CNN on the augmented images. We benchmark the performance of the CoDSA-enhanced model against a baseline CNN trained on the original dataset comprising 18,205 images, employing early termination based on validation set performance.

For synthetic image generation, we utilize Stable Diffusion 2.0 (\url{https://huggingface.co/stabilityai/stable-diffusion-2}), a latent diffusion model with an encoder and a decoder. We fine-tune it on a selected subset of 10,000 training images, updating only the diffusion modules while keeping the Variational Autoencoder (VAE) and text encoder unchanged.

Synthetic images are generated according to the conditional distribution $P(Z | Z \in C_k)$, for $k=1,\dots,5$, using Stable Diffusion. In this case, the split ratio is $r = \frac{10,000}{18,205}$. We combine the remaining 8,205 original images with $m$ synthetic images using the allocation vector $\bm{\alpha}$ as defined by ~\eqref{eq-alpha} for a balanced augmented data, varying $m$ from 10,000 to 25,000. We evaluate model performance by examining the root mean squared error (RMSE) for age predictions on the test set.

Experiments were repeated five times with independent random train–test splits; Figure~\ref{fig:da_face} shows the mean RMSE with standard error bars. CoDSA outperforms the baseline for $m=0$, but its error rises as $m$ increases. The best RMSE—7.952 at $m=10,000$ marks a substantial improvement over the baseline value of 9.81, attributable to adaptive synthetic augmentation.

\begin{figure}[!htbp]
\centering
\includegraphics[scale=0.5]{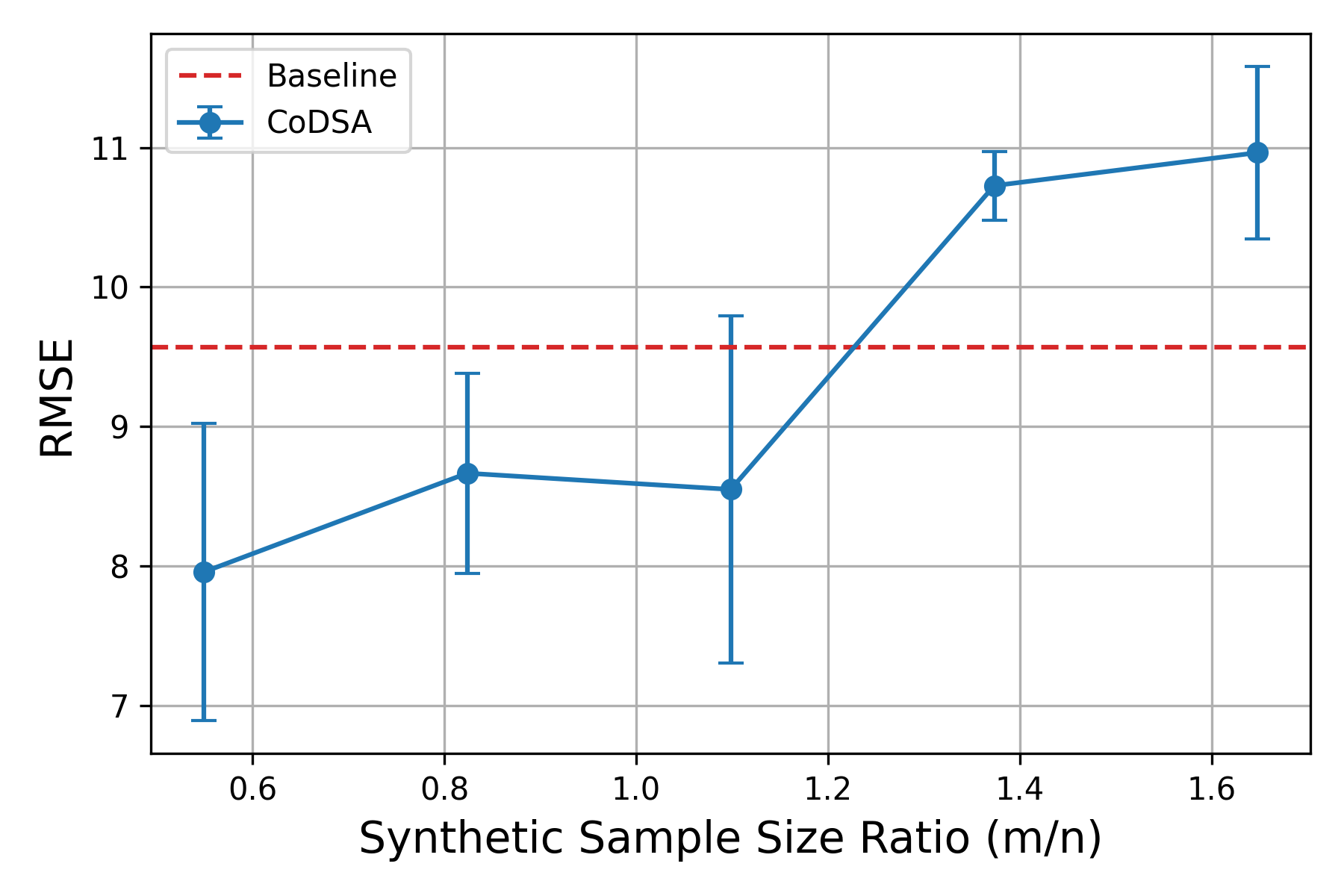}
\caption{Mean test root mean squared error ($\pm$ one standard error) for CoDSA on facial-age prediction across synthetic-to-real ratios $(m/n)$ with a split ratio $r = 10,000/18,205$, averaged over five random train–test splits and compared with a convolutional neural network (CNN) baseline.}
\label{fig:da_face}
\end{figure}

\subsection{Anomaly detection on MNIST-USPS data}

We evaluate the effectiveness of CoDSA on unsupervised tasks involving unstructured data using the MNIST-USPS benchmark dataset (\url{https://github.com/LukasHedegaard/mnist-usps}). This benchmark combines two handwritten digit datasets, MNIST and USPS, which differ in resolution, writing style, and digit variability. These discrepancies present a challenging setting for domain adaptation in anomaly detection.

In our anomaly detection setting, we define digit ``3'' images as normal, while digit ``5'' images are considered anomalous. We adopt an autoencoder-based detection method \citep{zhang2021towards}, which learns efficient representations by compressing and reconstructing input images. The autoencoder is exclusively trained on digit ``3'' images, enabling accurate reconstruction for normal digit ``3'' samples while producing high reconstruction errors for unseen anomalous digit ``5'' samples. The reconstruction error (loss) is 
$\ell(\x, \hat{\x}) = \|\x - \hat{\x}\|_2^2$,
where $\x$ denotes the original images,  $\hat{\x}$ denotes their corresponding reconstructions, and  $\ell$ serves as the anomaly score. Higher reconstruction loss indicates anomalous samples.

The autoencoder training set comprises 5,981 MNIST digit ``3'' images and 508 USPS digit ``3'' images, capturing variations of normal digits across both domains. For evaluation, we construct a balanced test set containing 150 USPS digit ``3'', 150 MNIST digit ``3'', 150 USPS digit ``5'', and 150 MNIST digit ``5'' images.

To enhance detection performance, we apply CoDSA to augment training images through synthetic image generation, enabling improved discrimination between normal and anomalous samples. Specifically, let $\Z$ denote images from either domain $\Z \in \{C_1, C_2\}$, representing MNIST and USPS, respectively. We generate synthetic images using a conditional diffusion model following $P(\Z| \Z\in C_k), k=1,2$. For MNIST-style digit ``3'' images, we utilize a pretrained conditional diffusion model \citep{TeaPearce}. Due to limited USPS data, this model is fine-tuned on USPS images (excluding digit ``3'' and digit ``5'' images) to realistically generate USPS-style digits, ensuring synthetic USPS images accurately reflect the dataset without direct replication. We select the allocation vector $\bm{\alpha}$ according to the optimal weight equation from \eqref{eq-alpha} with $q_k=0.5$, ensuring a balanced distribution of synthetic samples across domains and varying $m \in \{0, 6000, 12000\}$. The results measured in terms of AUC are shown in Table~\ref{tab:auc_results}.

\begin{table}[h]
\centering
\caption{Mean AUC scores (standard error) for CoDSA across different synthetic sample sizes $m$ over
five random train–test splits.}
\begin{tabular}{lccc}
\hline
& \textbf{USPS-AUC} & \textbf{MNIST-AUC} & \textbf{Overall AUC} \\
\hline
CoDSA ($m=0$) & 0.8321 (0.0209) & 0.8546 (0.0286) & 0.8266 (0.0180) \\
CoDSA ($m=6000$) & 0.8948 (0.0159) & 0.8949 (0.0220)& 0.8860 (0.0142)\\
CoDSA ($m=12000$) & 0.8957 (0.0150) & 0.9320 (0.0176)& 0.9014 (0.0111) \\
\hline
\end{tabular}
\label{tab:auc_results}
\end{table}

CoDSA significantly enhances detection performance. For $m=60,00$, the allocation vector $\bm \alpha=(\alpha_1,\alpha_2)$ is $(0.044, 0.956)$ for MNIST and USPS, substantially improving USPS performance while minimally affecting MNIST due to limited synthetic MNIST samples. Increasing to $m=12,000$ adjusts $\bm \alpha$ to $(0.272, 0.728)$, resulting in performance gains across both domains.

\section{Discussions}
\label{discussion}

This paper introduces Conditional Data Synthesis Augmentation (CoDSA), a novel framework designed for multimodal data augmentation. Our theoretical analysis demonstrates that CoDSA leverages conditional synthetic data generation to effectively balance underrepresented subpopulations and reduce statistical errors through adaptive sampling. A core insight of CoDSA involves two critical trade-offs. The first trade-off is the classical balance between approximation errors of the underlying function and estimation errors. The second trade-off involves balancing adaptive weighting effectiveness against generation errors arising during synthetic data generation.

Central to our approach is the generative model for producing high-fidelity synthetic samples that closely mimic the original data. This capability addresses class imbalance and undersampling issues within specific regions of interest. Through adaptive sampling, CoDSA effectively resolves several practical challenges, including undersampling in classification tasks, covariate shifts in regression scenarios, and particularly enhancing multimodal data augmentation—an area that remains largely unexplored.

Experimental evaluations across transfer learning tasks, including rating prediction from Yelp reviews, age prediction from facial images, and anomaly detection, demonstrate substantial improvements in predictive accuracy using adaptive synthetic augmentation. Specifically, combined transfer learning with adaptive synthetic data expansion yields robust performance gains, even with highly imbalanced data or data originating from disparate domains.

Additionally, our empirical findings highlight the importance of managing the synthetic sample size. While larger synthetic datasets typically enhance performance, excessive augmentation eventually encounters diminishing returns due to cumulative generation errors. This observation underscores the necessity of carefully selecting synthetic data volume and adaptive weights to achieve an optimal balance.

Future research avenues span both methodological and theoretical fronts:
(1) \textbf{Adaptive weighting and hyperparameter optimization.} It is of significance to refine adaptive weighting strategies and develop more efficient computational tools for identifying optimal hyperparameters within the CoDSA framework.
  (2) \textbf{Scaling pretraining data.} The ablation experiments in Supplementary Appendix D demonstrate that larger pretraining datasets consistently improve CoDSA’s stability and accuracy. Further empirical validation on diverse real‐world datasets will be essential to confirm these theoretical insights and to characterize how synthetic‐data quality interacts with transfer learning performance.
  (3) \textbf{Cross‐fitting enhancements.} The cross‐fitting study in Supplementary Appendix D indicates that, when computational resources permit, cross‐fitting can further stabilize CoDSA. Further work will involve both empirical investigations and theoretical analyses of cross-fitting strategies within this framework.
  (4) \textbf{Extension to inference tasks.} CoDSA’s core idea—leveraging synthetic data for downstream objectives—could be extended to statistical inference to improve estimator performance. A key challenge will be rigorously accounting for generative‐model error when drawing inference on target distributions.

\appendix 
\section{Conditional generation via diffusion models}
\label{transfer-diffusion}

A conditional diffusion model incorporates both forward and backward diffusion processes. For notational simplicity, we use $\zeta$ to denote the categorical condition that $\zeta=k$ if $\Z \in C_k$.

\noindent \textbf{Forward process. } The forward process systematically transforms a random vector $\U(0)$ from the conditional distribution of $\U$ given the condition $\zeta$ into pure white noise by progressively
injecting white noise into a differential equation defined with the Ornstein-Uhlenbeck process, leading to diffused distributions from the initial state $\U(0)$:
\begin{equation}
\label{forward}
\mathrm{d}\U(t)=-{b}_{t} \U(t)\mathrm{d}t+\sqrt{2{b}_t}\mathrm{d}W(t),\quad t \geq 0,
\end{equation}
where $\U(0)$ follows a conditional probability density $p_{\bm{u}|\zeta}$, $\{W(t)\}_{t\geq 0}$ represents a standard Wiener process and ${b}_t$ is a non-decreasing weight function. Under \eqref{forward}, $\U(t)$ given $\U(0)$ follows $N(\mu_{t}\U(0),\sigma^2_{t}\bm I)$, where $\mu_{t}=\exp(-\int_0^t {b}_s\mathrm{d} s)$ and $\sigma^2_{t}=1-\mu_{t}^2$. Here, we set
${b}_s=1$, which results in $\mu_{t}=\exp{(-t)}$ and $\sigma^2_{t}=1-\exp{(-2t)}$. Practically, the process terminates at a sufficiently large $\overline{t}$, ensuring the distribution of $\U(t)$,
a mixture of $\U(0)$ and white noise, resembles the standard Gaussian vector.

\noindent \textbf{Backward process.} Given $\U(\overline{t})$ in \eqref{forward}, a backward process is employed for sample generation for $\U(0)$. Assuming \eqref{forward} satisfies certain conditions \citep{anderson1982reverse}, the backward process $\bm V(t)=\U(\overline{t}-t)$, starting with $\U(\overline{t})$, is derived as:
\begin{equation}
\label{reverse}
\mathrm{d}\bm V(t)={b}_{\overline{t}-t}(\bm V(t)+2\nabla\log p_
{\bm{u}(\overline{t}-t)}(\U(\overline{t}-t))\mathrm{d}t+\sqrt{2{b}_{\overline{t}-t}}\mathrm{d}W(t); \quad t \geq 0,
\end{equation}
where $\nabla\log p_{\bm u}$ is the score function which represents the gradient of $\log p_{\bm u}$.

\noindent \textbf{Conditional score matching.} To estimate the unknown score function, we minimize a matching loss between the score and its approximator $\theta_u$:
$\int_{0}^{\overline{t}}\mathrm{E}_{\bm{u}(t)}\|\nabla \log p_{\bm{u}(t)|\zeta}(\U(t))-\theta_u(\U(t),\zeta,t)\|^2\mathrm{d}t$,
where $\|\x\|=\sqrt{\sum^{d_x}_{j=1}\x_j^2}$ is the Euclidean norm, which is equivalent to minimizing the following loss \cite{oko2023diffusion},
\begin{equation}
\label{loss_2}
\int_{\underline{t}}^{\overline{t}}\mathrm{E}_{\bm u(0),\zeta}\mathrm{E}_{\bm{u}(t)|\bm u(0)}\|\nabla \log p_{\bm{u}(t)|\bm u(0)}(\U(t)|\U(0))-\theta(\U(t),\zeta,t)\|^2\mathrm{d}t,
\end{equation}
with $\underline{t}=0$. In practice, to avoid score explosion due to $\nabla \log p_{\bm{u}(t)|\bm u(0)} \rightarrow \infty$ as $t\rightarrow 0$ and to ensure training stability, we restrict the integral interval to $\underline{t}>0$ \citep{oko2023diffusion,chen2023improved} in the loss function.
Then, both the integral and $\mathrm{E}_{\bm{u}(t)|\bm u(0)}$ can be precisely approximated by sampling $t$ from a uniform distribution on $[\underline{t},\overline{t}]$ and a sample of $\U(0)$ from the conditional distribution of $\U(t)$ given $\U(0)$. Hence, for a conditional generation of $\U$, we use the diffusion model to learn the latent distribution of $\U$ from the training data $\{\bm{u}^i\}_{i=1}^{n_{g}}$, with $\U(0)=\U$ in
\eqref{forward}. Based on a training sample $\{\bm{u}^i,\zeta^i\}_{i=1}^{n_{g}}$, the empirical score matching loss in \eqref{loss_2} $L_{u}$ is 
$$L_{u}(\theta_u)
=\sum_{i=1}^{n_{g}} \int_{\underline{t}}^{\overline{t}} \mathrm{E}_{\bm{u}(t)|\bm u(0)}\|\nabla \log p_{\bm{u}(t)|\bm u(0)}(\U(t)|
\U^i)-\theta_u(\U(t),\zeta^i,t)\|^2 \mathrm{d}t.$$

\noindent \textbf{Neural network. } An $\mathbb{L}$-layer network $\Phi$ is defined by a composite function
$
\Phi(\x)=(\bm{\mathrm{A}}_\mathbb{L}\sigma(\cdot)+\bm{b}_\mathbb{L})\circ\cdots(\bm{\mathrm{A}}_2\sigma(\cdot)+\bm{b}_2)\circ (\bm{\mathrm{A}}_1\x+\bm{b}_1),
$
where $\bm{\mathrm{A}}_i\in \R^{d_{i+1}\times d_i}$ is a weight matrix and $\bm b_i \in \R^{d_{i+1}}$ is the bias of a linear transformation of the $i$-th layer, and $\sigma$ is the ReLU activation function, defined as $\sigma(\x)=\max(\x,0)$.
Then, the parameter space $\Theta$ is set as $\mathrm{NN}(d_{in},d_{out},\mathbb{L},\mathbb{W},\mathbb{S},\mathbb{B},\E)$ 
with $\mathbb{L}$ layers, a maximum width of $\mathbb{W}$, effective parameter number $\mathbb{S}$, the sup-norm $\mathbb{B}$, and parameter bound $\E$:
\begin{align}
\label{p-space}
\mathrm{NN}(d_{in},d_{out},\mathbb{L},\mathbb{W},\mathbb{S},\mathbb{B},\E)
=&\{ \Phi: d_1=d_{in}, d_{\mathbb{L}+1}=d_{out}, \max_{1\leq i\leq \mathbb{L}}d_i\leq\mathbb{W}, \nonumber\\
&\sum_{i=1}^{\mathbb{L}}(\|\bm{\mathrm{A}}_i\|_0+ \|\bm{b}_i\|_0)\leq \mathbb{S}, 
\|\Phi\|_{\infty}\leq\mathbb{B},\nonumber\\ 
&\max_{1\leq i\leq \mathbb{L}}(\|\bm{\mathrm{A}}_i\|_{\infty}, \|\bm{b}_i\|_{\infty})\leq \E
\},
\end{align}
where $\|\cdot\|_{\infty}$ is the maximal magnitude of entries and $\|\cdot\|_0$ is the number of nonzero entries.

With $\theta_u\in \Theta_{u}=\mathrm{NN}(d_u+1,d_u,\mathbb{L}_u,\mathbb{W}_u,\mathbb{S}_u,\mathbb{B}_u,\E_u)$, the estimated score function for the target task is
$\hat{\theta}_u=\argmin_{\theta_u\in\Theta_{U}}L_{U}(\theta_u)$.
Finally, $P_{\z|\zeta}$ is estimated by $\hat{P}_{\z|\zeta}=\hat{P}_{u|\zeta}(g^{-1})$.

\noindent \textbf{Generation.} To generate a random sample of $\bm V(t)$, we replace the score $\nabla\log p_{\bm{u}(\overline{t}-t)}$ by its estimate $\hat \theta$ in \eqref{reverse} to yield $\bm V(t)$ in the backward equation. For implementation, we may utilize a discrete-time approximation of the sampling process, facilitated by numerical methods for solving stochastic differential equations, such as Euler-Maruyama and stochastic Runge-Kutta methods \cite{song2020denoising}.

\begin{lemma}
\label{thm_diff_general}
Under Assumption \ref{A_ex}, setting the neural network's structural hyperparameters of $\Theta_u=\mathrm{NN}(\mathbb{L},\mathbb{W},\mathbb{S},\mathbb{B},\E)$ as follows: $\mathbb{L}=c_L \log^4\mathcal{K}$, $\mathbb{W}= c_W K\mathcal{K}\log^7\mathcal{K}$, $\mathbb{S}= c_S \mathcal{K}\log^9\mathcal{K}$, 
$\log\mathbb{B}= c_B\log \mathcal{K}$, $\log\E= c_E\log^4\mathcal{K}$, with diffusion stopping criteria from \eqref{forward}-\eqref{reverse} 
as $\log \underline{t}= -c_{\underline{t}}\log \mathcal{K}$ and $\overline{t}=c_{\overline{t}}\log \mathcal{K}$, where $\{c_L,c_W, c_S, c_B,c_E,c_{\underline{t}},c_{\overline{t}}\}$ are sufficiently large constants, yields the error in diffusion generation via transfer learning: 
\begin{eqnarray}
\label{c-rate-general}
\mathrm{E}_{\zeta}[\mathrm{TV}(p^0_{\bm u|\zeta},\tilde{p}_{\bm u|\zeta})]\leq \beta_{n_{g}} + \delta_{n_{g}}, 
\end{eqnarray}
with probability exceeding $1-\exp(-c_u n^{1-\xi}_{n_{g}} (\beta_{n_{g}} + \delta_{n_{g}})^2)$ for
some constant $c_u>0$ and a small $\xi>0$. Here, $\beta_{n_{g}}$ and $\delta_{n_{g}}$ represent the estimation and approximation errors, given by:
$ \beta_{n_{g}}\asymp 
\sqrt{\frac{\mathcal{K}\log^{19}\mathcal{K}}{n_{g}}},\quad \delta_{n_{g}}\asymp \mathcal{K}^{-\frac{r}{d_x+d_z}}\log^{\frac{r}{2}+1}\mathcal{K}$.
Setting $\beta_{n_{g}}= \delta_{n_{g}}$ to solve for $\mathcal{K}$, and neglecting the logarithmic term, 
leads to 
$\beta_{n_{g}} = \delta_{n_{g}} \asymp n_{g}^{-\frac{r_{u}}{d_u+2r_{u}}}\log^m n_{g}$ with 
the optimal $\mathcal{K} \asymp n_{g}^{\frac{d_u}{d_u+2r_{u}}}$, with $\beta_u = \max(\frac{19}{2}, \frac{r_u}{2}+1)$. 
\end{lemma}
\begin{proof}[Proof of Lemma \ref{thm_diff_general}]
The result follows from the same arguments as those in Theorem 2 of \cite{tian2024enhancing} and Theorem 4.1 of \cite{fu2024unveil}. For each group \(k\), we employ the same approximation neural network architecture as in these studies to approximate the corresponding conditional score function. In the final step, we append additional layers that approximate the multiplicative interactions between the output of each group-specific network and the one-hot encoded vector \(\zeta\). This modification yields the final output for group \(\zeta\).

This adjustment increases the overall network width by a factor of $K$ compared to the original neural network classes in \cite{tian2024enhancing,fu2024unveil}, since we use $K$ parallel sub-networks, one for each group, to capture the conditional score function. Aside from this increase in the approximation network's width, all other components of the proof remain identical to those in the referenced works.
\end{proof}

\noindent \textbf{Transfer diffusion models.} A common strategy to improve the accuracy of conditional diffusion generation on a target task is to fine-tune a diffusion model that has been pretrained on large-scale source data. This approach leverages the knowledge learned during pretraining and facilitates effective transfer to the target domain. Fine-tuning strategies generally fall into two categories: full-parameter fine-tuning, where all parameters of the pretrained model are updated, and partial fine-tuning, where only a subset of parameters, such as attention layers or the noise prediction head, are updated, while others (e.g., the encoder-decoder backbone) remain frozen, as shown in Figure \ref{fig:da_lda}. This flexibility enables adaptation to new tasks while mitigating overfitting and reducing computational cost.

\begin{figure}[!htbp]
\centering
\includegraphics[scale=0.25]{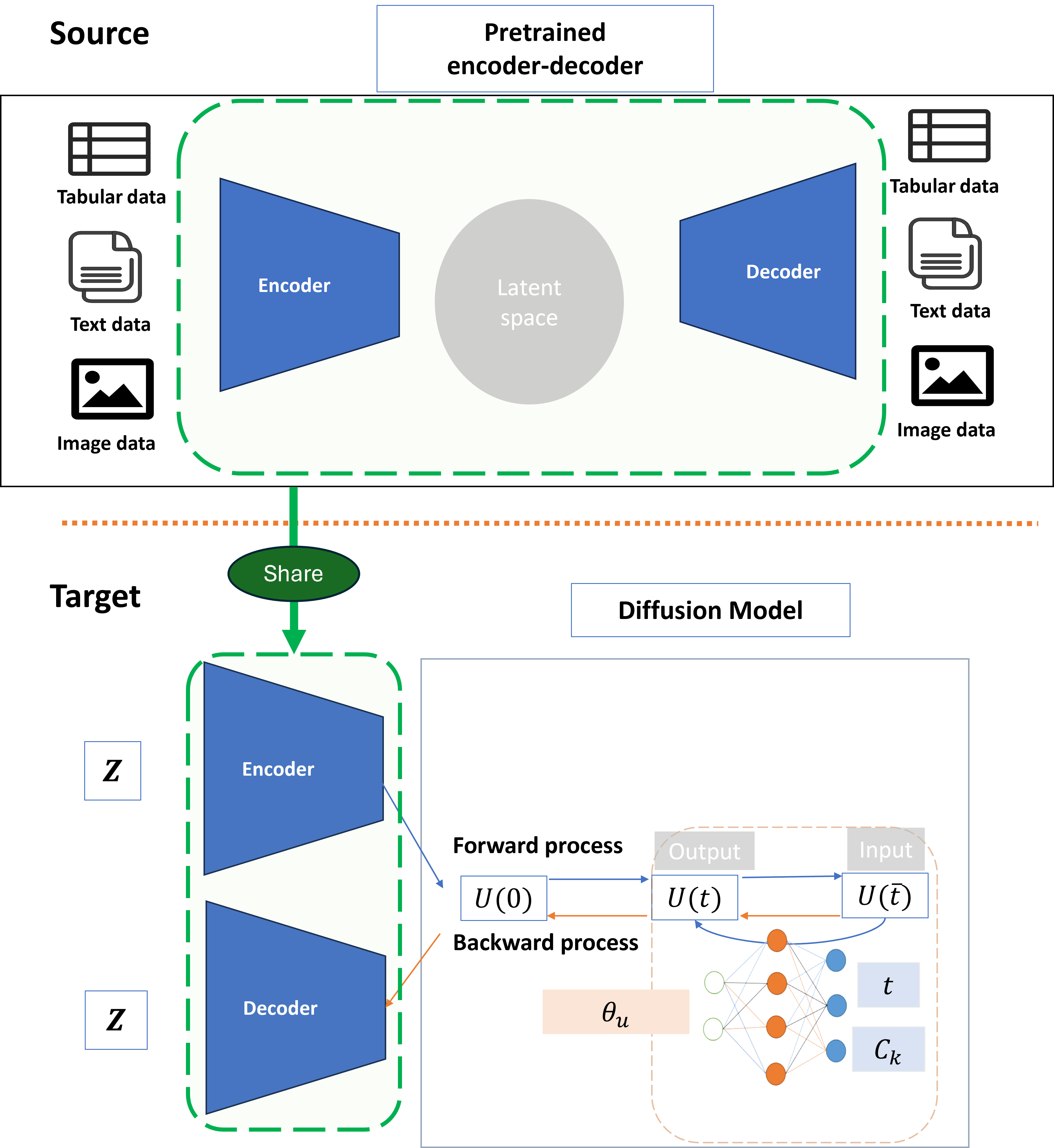}
\caption{Latent diffusion model with pretrained encoder-decoder from source dataset.}
\label{fig:da_lda}
\end{figure}

\begin{prop}[Conditional diffusion via transfer learning]
\label{thm_ug_detail}
Under Assumption \ref{A_ex}, there exists a ReLU network class such that the error for unconditional diffusion generation via transfer learning, measured in Wasserstein distance, satisfies:
$
\mathcal{W}(P_k, \tilde{P}_k) \leq \frac{\varepsilon_u}{p_k}+\varepsilon_s,
$
with probability at least \(1 - \exp(-c_t n_{g}^{1-\xi}\varepsilon_u^2)\) for some constant \(c_t > 0\) and a small $\xi>0$. Here,
$
\varepsilon_u \asymp n_{g}^{-\frac{r_u}{d_u+r_u}} \log^{\beta_u}(n_{g})$, where $\beta_u = \max\left(\frac{19}{2}, \frac{r_u}{2}+1\right)$.
\end{prop}

\begin{proof}[Proof of Proposition \ref{thm_ug_detail}]
Let \(\bar{P}_{\z|\zeta}\) denote the distribution of the reconstructed target variable \(\bar{\Z} = g(\bm{U})\), with latent representation \(\bm{U} = f(\Z)\). By the triangle inequality for Wasserstein distances, we have the following decomposition:
$W\bigl(P_{\z|\zeta}, \tilde{P}_{\z|\zeta}\bigr)
\leq W\bigl(\bar{P}_{\z|\zeta}, \tilde{P}_{\z|\zeta}\bigr)
+ W\bigl(P_{\z|\zeta}, \bar{P}_{\z|\zeta}\bigr)$.

The first term on the right-hand side captures the error introduced by the diffusion model in latent space, while the second term represents the reconstruction error due to the mapping \(g\).

Using the boundedness of $g$, we obtain
$W\bigl(\bar{P}_{\z|\zeta}, \tilde{P}_{\z|\zeta}\bigr)
\leq B_x \, \mathrm{TV}\bigl(\bar{P}_{\z|\zeta}, \tilde{P}_{\z|\zeta}\bigr)
\leq  B_x\mathrm{TV}\bigl(\bar{P}_{\bm{u}|\zeta}, \tilde{P}_{\bm{u}|\zeta}\bigr)$,
where \(\bar{P}_{\bm{u}|\zeta}\) and \(\tilde{P}_{\bm{u}|\zeta}\) denote the true and estimated distributions of the latent variable \(\bm{u}\).

Under Assumption \ref{A_ex} and applying Lemma \ref{thm_diff_general}, we have
$
\E_{\zeta}\bigl[W\bigl(\bar{P}_{\bm{u}|\zeta}, \hat{P}_{\bm{u}|\zeta}\bigr)\bigr] \leq \varepsilon_u,
$
with probability at least \(1 - \exp(-c_t n_{g}^{1-\xi}\varepsilon_u^2)\). Decompose the expectation with respect to the discrete $\zeta$, we get the bound of $\frac{\varepsilon_u}{p_k}$.

For the reconstruction term, by the definitions of Wasserstein distance and the reconstruction error,
$W\bigl(P^0_{\z|\zeta}, \bar{P}_{\z|\zeta}\bigr)
\leq \E\bigl[\|g(f(\z)) - \z\|\bigr]
\leq \varepsilon_s$.
Combining these bounds completes the proof.
\end{proof}

\section{Proofs}
\label{sec:proof}

 We first state Lemma 11 from \citep{tian2024enhancing}, which will be utilized in the subsequent proofs.

\begin{lemma}[Lemma 11 from \cite{tian2024enhancing}]
\label{large-d}
Suppose the class of functions \(\mathcal{F}\) satisfies the Bernstein-type condition for independent (but not necessarily identically distributed) random variables \(\Z_1, \dots, \Z_m\), i.e., for some constant \(U > 0\), we have
\[
\E|f(\Z_j)-\E f(\Z_j)|^k\leq \frac{1}{2}k! v_j U^{k-2}, \quad j=1,\dots,m,\quad k\geq 2.
\]
Define
$\varphi(\varepsilon_m,v,U) = \frac{m \varepsilon_m^2}{2\left[4v + \frac{\varepsilon_m U}{3}\right]}$,
where \( v \geq \frac{1}{m}\sum_{j=1}^{m} v_j \geq \frac{1}{m}\sum_{j=1}^{m}\mathrm{Var}(f(\Z_j)) \). Assume the following conditions hold:
$\varepsilon_m \leq \frac{\gamma v}{4U}$,
\begin{align}
\int_{\gamma \varepsilon_m/8}^{\sqrt{v}} H^{1/2}(u,\mathcal{F})\,du &\leq \frac{\sqrt{m}\,\varepsilon_m \gamma^{3/2}}{2^{10}}, \quad \text{for some } 0<\gamma<1,
\end{align}
where $H(\cdot, \mathcal{F})$ is the bracketing $L_2$-metric entropy for the class \(\mathcal{F}= \{\frac{1}{m}\sum_{j=1}^mf(\Z_j, \theta): \theta \in \Theta\}\). Then, 
\[
\mathbb{P}^{*}\left(\sup_{f\in\mathcal{F}} \frac{1}{m}\sum_{j=1}^{m}\bigl(f(\Z_j)-\E f(\Z_j)\bigr)\geq 
\varepsilon_m\right) \leq 3\exp\left(-(1-\gamma)\varphi(\varepsilon_m,v,U)\right),
\]
where $\mathbb{P}^{*}$ denotes the outer probability measure associated with the independent samples \(\Z_1,\dots,\Z_m\).
\end{lemma}

\begin{proof}[Proof of Theorem \ref{thm1}]
The proof is organized into six steps.

\medskip
\textbf{Step 1: Notation.} \\
We first introduce the notation for the various risks. Define the evaluation risk by
$R_Q = \mathbb{E}_Q\, l(\theta,\Z)$.
The CoDSA empirical loss defined in \eqref{CoDSA1} is  
$L_{m,n_r}(\theta,\Z_{c}) = \frac{1}{m+n_{r}} \sum_{k=1}^K\bigl( \sum_{j=1}^{m_k} \ell(\theta, \tilde \Z_{kj})
+\sum_{j=1}^{n_{r,k}} \ell(\theta, \Z_{kj})\bigr)$.
Its corresponding risk is
$R_C(\theta) = \mathbb{E}\, L_{m,n_r}(\theta,\Z_c)$.
We denote by
$\theta_C = \arg\min_{\theta \in \Theta} R_C(\theta)$
the minimizer of $R_C$, and define the excess risk relative to $\theta_C$ as
$\rho_C(\theta,\theta_C) = R_C(\theta) - R_C(\theta_C)$.

\medskip
\textbf{Step 2: Bounding the discrepancy between $R_C$ and $R_Q$.} \\
We express the difference as
\begin{align}
\label{eq:discrepancy}
   &R_C(\theta)-R_Q(\theta)\nonumber\\
=& \sum_{k=1}^K\left(\frac{m_k}{m+n_{r}}\E_{\tilde \Z|C_k}\ell(\theta,\tilde\Z)+\frac{n_{r,k}}{m+n_{r}}\E_{\Z|C_k}\ell(\theta,\Z)
\right)-\sum_{k=1}^K\tilde{\alpha}_k\E_{\Z|C_k}\ell(\theta,\Z)
\nonumber\\
&+\sum_{k=1}^K\tilde{\alpha}_k\E_{\Z|C_k}\ell(\theta,\Z)-\sum_{k=1}^K q_k\E_{\Z|C_k}\ell(\theta,\Z)\nonumber\\
   \leq &\sum_{k=1}^K\frac{m_k}{m+n_r}\Biggl|\E_{\Z|C_k} \ell(\theta,\Z)-\E_{\tilde{\Z}|C_k}\ell(\theta,\tilde{\Z})\Biggr| +  \Biggl|\sum_{k=1}^K (\tilde\alpha_k-q_k)\,\E_{\Z|C_k}\ell(\theta,\Z )\Biggr|
\end{align}
where 
$\tilde{\alpha}_k = \frac{m_k+n_{r,k}}{m+n_r} = \frac{\alpha_k m+p_k(1-r)n}{m+(1-r)n}$.
Since $\ell$ is $\beta$-Lipschitz, by the Kantorovich--Rubinstein duality \citep{kantorovich1958duality}, for each $k$ we have
\begin{equation}
\label{eq:lip-bound}
\Bigl|\mathbb{E}_{\Z|C_k}\,\ell(\theta,\Z)
-\mathbb{E}_{\tilde{\Z}|C_k}\,\ell(\theta,\tilde{\Z)}\Bigr|
\le \beta\, \operatorname{W}(P_k,\hat{P}_k).
\end{equation}
Moreover, using the uniform bound $U$ for $\ell$, the second term in \eqref{eq:discrepancy} satisfies
\begin{equation}
\label{eq:weight-bound}
\Biggl|\sum_{k=1}^K (\tilde{\alpha}_k-q_k)\,\mathbb{E}_{\Z|C_k}\,\ell(\theta,\Z)\Biggr|
\le U\, D_{\lambda},
\end{equation}
where
$D_\lambda=\sum_{k=1}^K |\tilde{\alpha}_k-q_k|$.
Combining \eqref{eq:lip-bound} and \eqref{eq:weight-bound} in \eqref{eq:discrepancy} yields
\begin{eqnarray}
\label{b2}
|R_C(\theta)-R_Q(\theta)|\le \beta\, G_{\lambda}+ U\, D_{\lambda},
\end{eqnarray}
with
$G_{\lambda} = \sum_{k=1}^K \frac{\alpha_k}{1+(1-r)n/m}\,\tau_{k,n_g}$.

\medskip
\textbf{Step 3: Decomposition of the excess risk.} \\
By \eqref{b2}, the overall excess risk can be decomposed as follows:
\begin{align}
\label{eq:excess-risk-decomp}
\rho(\hat{\theta}_\lambda,\theta_0)
&= R_Q(\hat{\theta}_\lambda)-R_Q(\theta_0) \nonumber\\[1mm]
&=\Bigl[R_Q(\hat{\theta}_\lambda)-R_C(\hat{\theta}_\lambda)\Bigr]+\Bigl[R_C(\hat{\theta}_\lambda)-R_C(\theta_C)\Bigr]
+\Bigl[R_C(\theta_C)-R_Q(\theta_0)\Bigr]\nonumber\\[1mm]
&\le \rho_C(\hat{\theta}_\lambda,\theta_C)
+ U\, D_{\lambda} + \beta\, G_{\lambda}
+\Bigl[R_C(\theta_C)-R_Q(\theta_0)\Bigr].
\end{align}
Next, decompose
\begin{align}
\label{eq:approx-error}
R_C(\theta_C)-R_Q(\theta_0)
&=\Bigl[R_C(\theta_C)-R_C(\theta_Q)\Bigr]
+\Bigl[R_C(\theta_Q)-R_Q(\theta_Q)\Bigr] \nonumber\\[1mm]
&\quad+\Bigl[R_Q(\theta_Q)-R_Q(\theta_0)\Bigr] \nonumber\\[1mm]
&\le U\, D_{\lambda}+ \Bigl(R_Q(\theta_Q)-R_Q(\theta_0)\Bigr),
\end{align}
since $\theta_C$ minimizes $R_C$, so that $R_C(\theta_C)-R_C(\theta_Q) \le 0$. Plugging \eqref{eq:approx-error} into \eqref{eq:excess-risk-decomp} yields
\[
\rho(\hat{\theta}_\lambda,\theta_0)
\le \rho_C(\hat{\theta}_\lambda,\theta_C)
+ \Bigl(R_Q(\theta_Q)-R_Q(\theta_0)\Bigr)
+ 2U\, D_{\lambda} + 2\beta\, G_{\lambda}.
\]
Thus, if we choose
\[
\varepsilon > \max\{\rho(\theta_Q,\theta_0),\, D_{\lambda},\, G_{\lambda}\},
\]
then
$\rho(\hat{\theta}_\lambda,\theta_0) \le \rho_C(\hat{\theta}_\lambda,\theta_C) + (c_e-1)\,\varepsilon$,
with $c_e = 2+2U+2\beta$.
This implies that
\[
\{\rho(\hat{\theta}_\lambda,\theta_0) \ge c_e\,\varepsilon\} \subseteq \{\rho_C(\hat{\theta}_\lambda,\theta_C) \ge \varepsilon\}.
\]

\medskip
\textbf{Step 4: Empirical process bound.} \\
Partition the parameter space into slices
\[
A_l = \Bigl\{\theta\in\Theta:\, 2^l\varepsilon < \rho_C(\theta,\theta_C) \le 2^{l+1}\varepsilon\Bigr\},\quad l=0,1,\ldots.
\]
By the definition of $\hat{\theta}_\lambda$, we have
\[
P\Bigl(\rho_C(\hat{\theta}_\lambda,\theta_C) \ge \varepsilon\Bigr)
\le \sum_{l=0}^{\infty} P\Bigl(\sup_{\theta\in A_l} \nu\Bigl(L_{m,n_r}(\theta,\Z_c)-L_{m,n_r}(\theta_C,\Z_c)\Bigr) > 2^l\varepsilon\Bigr),
\]
where $\nu$ denotes the empirical process based on $\Z_c$.

\medskip
\textbf{Step 5: Bounding the variance.} \\
For $\theta\in A_l$, note that by independence, the variance term can be expressed as,
\[
V(\theta,\theta_C)=\frac{1}{m+n_r}\sum_{k=1}^K\left[ m_k\,\operatorname{Var}_{\tilde{\Z}|C_k}\Bigl(\ell(\theta,\Z)-\ell(\theta_C,\Z)\Bigr)
+ n_{r,k}\,\operatorname{Var}_{\Z|C_k}\Bigl(\ell(\theta,\Z)-\ell(\theta_C,\Z)\Bigr)\right].
\]
First, we have $V(\theta,\theta_C)\leq V(\theta,\theta_0) + V(\theta_C,\theta_0)$, where $V(\theta,\theta_0)$ is defined in the same form of $V(\theta,\theta_C)$ replacing $\theta_C$ by $\theta_0$.
By the $\beta$-Lipschitz property and the uniform bound $U$, we obtain
\begin{equation}
\label{eq:var-bound}
\Bigl|\operatorname{Var}_{\tilde{\Z}|C_k}\Bigl(\ell(\theta,\Z)-\ell(\theta_0,\Z)\Bigr)
-\operatorname{Var}_{\Z|C_k}\Bigl(\ell(\theta,\Z)-\ell(\theta_0,\Z)\Bigr)\Bigr|
\le 4U\beta\, \operatorname{W}(P_k,\tilde{P}_k).
\end{equation}
Using a decomposition analogous to that in Step 2, it follows that
\begin{align*}
V(\theta,\theta_0)
&\le \sum_{k=1}^K \Biggl[ \tilde{\alpha}_k\,\operatorname{Var}_{\Z|C_k}\Bigl(\ell(\theta,\Z)-\ell(\theta_0,\Z)\Bigr)
+ 4U\beta\, \frac{m_k}{m+n_r}\,\operatorname{W}(P_k,\tilde{P}_k) \Bigg]\\[1mm]
&\le \operatorname{Var}_Q\Bigl(\ell(\theta,\Z)-\ell(\theta_0,\Z)\Bigr)
+ 4U\beta\sum_{k=1}^K \frac{m_k}{m+n_r}\,\operatorname{W}(P_k,\tilde{P}_k)
+ 4U^2\sum_{k=1}^K |q_k-\alpha_k|\\[1mm]
&\le c_v\,\rho(\theta,\theta_0)
+ 4U\beta\sum_{k=1}^K \frac{m_k}{m+n_r}\,\operatorname{W}(P_k,\tilde{P}_k)
+ 4U^2\sum_{k=1}^K |q_k-\alpha_k|\\[1mm]
&\le c_v\,\rho_C(\theta,\theta_C)
+ \Bigl(c_v\,c_e+4U^2+4U\beta\Bigr)\,\varepsilon,
\end{align*}
where we have assumed $\varepsilon \ge \max\{D_{\lambda},\, G_{\lambda}\}$. 
Then for $\theta\in A_l$, we have 
\begin{align*}
V(\theta,\theta_C)\leq c_v2^{l+1}\varepsilon+2\Bigl(c_v\,c_e+4U^2+4U\beta\Bigr)\,\varepsilon.
\end{align*}
 
Then, for each $l$, using the bound above for $v$ in Lemma~\ref{large-d}, we obtain
\[
\mathbb{P}\Bigl(\sup_{\theta\in A_l} \nu\Bigl(L_{m,n_r}(\theta,\Z_c)-L_{m,n_r}(\theta_C,\Z_c)\Bigr) > 2^l\varepsilon\Bigr)
\le 3\exp\Biggl(-\frac{(1-\gamma)}{c_d}\, (m+n_r)\,2^l\varepsilon\Biggr),
\]
where 
\[
c_d = 2\Biggl[4\,\Bigl(2c_v(c_e+1)+8U(\beta+U)\Bigr)
+\frac{U}{3}\Biggr].
\]

Summing over $l\ge0$, we deduce that
\begin{equation}
\label{eq:empirical-bound}
\mathbb{P}\Bigl(\rho_C(\hat{\theta}_\lambda,\theta_C) \ge \varepsilon\Bigr)
\le 4\exp\Biggl(-\frac{(1-\gamma)}{c_d}\, (m+n_r)\,\varepsilon\Biggr).
\end{equation}

\medskip
\textbf{Step 6: Incorporating the generation error.} \\
Define the event
$\mathcal{S} = \Bigl\{\operatorname{W}(P_k,\tilde{P}_k) \le \tau_{k,n_g},\quad k=1,\ldots,K\Bigr\}$.
By Assumption~\ref{A-3} and applying a union bound, we have
\[
\mathbb{P}\bigl(\mathcal{S}^c\bigr) \le \sum_{k=1}^K \exp\Bigl(-c_g\, n_g\, \tau_{k,n_g}\Bigr).
\]
Conditioning on $\mathcal{S}$ and combining with \eqref{eq:empirical-bound}, it follows that with probability at least
\[
1-\sum_{k=1}^K \exp\Bigl(-c_g\, n_g\, \tau_{k,n_g}\Bigr)-4\exp\Biggl(-\frac{(1-\gamma)}{c_d}\, (m+n_r)\,\varepsilon\Biggr),
\]
we have
$\rho(\hat{\theta}_\lambda,\theta_0) \le c_e\,\varepsilon$.
This completes the proof.
\end{proof}

\begin{proof}[Proof of Theorem \ref{thm2}]
Theorem \ref{thm2} is a direct result combining the results of Lemma \ref{thm_diff_general}, Lemma \ref{l_classification}, and Theorem \ref{thm1}. 
\end{proof}

\begin{lemma}[Classification error]
\label{l_classification}
When Assumption \ref{A_f} holds, if we set the classifier as $\Theta=\mathrm{NN}(d_x,1,\mathbb{W},\mathbb{L},\mathbb{S},\mathbb{B},\E)$ with $\mathbb{W}=c_W(W\log W)$, $\mathbb{L}=c_L(L\log L)$, $\mathbb{S}=c_S\mathbb{W}^2\mathbb{L}$, $\mathbb{B}=B_x$ and $\E=(WL)^{c_E}$ with some positive constants $c_W$, $c_L$, $c_S$ and $c_E$ dependent on $d_x$ and $r_{\theta}$, it holds that $
\rho(\theta_Q,\theta_0)=O((WL)^{-\frac{2r_{\theta}}{d_x}})
$ and $\delta_{m+n_{r}}=O(\sqrt{\frac{(WL)^2\log^5 (WL)}{m+n_{r}}})$. If we set $\rho(\theta_Q,\theta_0)=\delta_{m+n_{r}}$ to determine $WL$,
after ignoring the logarithmic term, the optimal bound is obtained by $WL\asymp (m+n_r)^{\frac{d_x}{2(d_x+2{r_{\theta}})}}$. This yields
$\xi_{m+n_{r}} =O((m+n_r)^{-\frac{r_{\theta}}{d_x+2r_{\theta}}} \log^{\beta_{\theta}} (m+n_r))$, where $\beta_{\theta}=\frac{5}{2}$. 
\end{lemma}

\begin{proof}
Since $\theta_0$ is well bounded from Assumption \ref{A_f}, we can bound the $\rho(\theta_Q,\theta_0)$ by the sup-norm of $\theta_Q-\theta_0$, $\rho(\theta_Q,\theta_0)\leq c \|\theta_Q-\theta_0\|_{\infty}$, with some positive constant $c$. 
The first term $\rho(\theta_Q,\theta_0)$ is the approximation error from the approximation theory (Theorem 1.1 in \cite{lu2021deep} and Lemma 11 in \cite{huang2022error}) of neural networks for the smooth function class. The second entropy bound is from the Lemma C.2 \citep{oko2023diffusion} and Lemma D.8 \citep{fu2024unveil}, $H(\cdot,\Theta)$
is bounded by the hyperparameters of depth $\mathbb{L}$, width $\mathbb{W}$, number of parameters $\mathbb{S}$, parameter bound $\E$ and the diameter of the input domain, $H(u,\Theta)= O(\mathbb{S}\mathbb{L}\log(\E\mathbb{W}\mathbb{L}/u))$. Plug in the parameters of $\Theta$ and solve the entropy inequality, so we obtain $\delta_{m+n_{r}}$. This completes the proof.
\end{proof}

\section{Experiment details}
\label{exp-details}

This section provides additional details regarding the experimental setups outlined in Section~\ref{sec:numerical}.

\noindent \textbf{Simulation: Imbalanced Classification.} For both of the SMOTE and ADASYN methods, the tuning parameters include the synthetic sample ratio $(m/n)$ and the minority group ratio $(\alpha)$. For the two CoDSA methods, an additional tuning parameter, the split ratio $(r)$, is considered. Parameter optimization is conducted via grid search to minimize cross-entropy loss on the validation set. The grid comprises:
\[
\{(r, \alpha, \tfrac{m}{n}) : r \in \{0.1, 0.2, \dotsc, 1\},\; \alpha \in \{0.1, 0.2, \dotsc, 0.9\},\; \tfrac{m}{n} \in \{0.1, 0.2, \dotsc, 2\}\}.
\]

The encoder and decoder are implemented as a three-layer neural network with ReLU activations, 256 hidden units per layer, and a latent dimension of 3, trained using squared loss. In the diffusion model, $\theta_u$ is set as a 10-layer neural network with 1024 hidden neurons in each layer. Time embedding transforms timesteps into a 128-dimensional representation using a two-layer neural network. The diffusion process comprises 1000 steps, with noise levels ranging between $0.0001$ and $0.02$. The Adam optimizer is used to train the diffusion model for 5000 epochs with a learning rate of $0.0001$. The classification baseline employs a three-layer neural network with ReLU activations in the first two layers, each with 128 hidden neurons.

\noindent \textbf{Simulation: Regression with undersampled regions.} 
Parameter optimization employs grid search to minimize the cross-entropy loss on the validation set, using the identical parameter grid as below:
$
\{(r, \alpha, \tfrac{m}{n}) : r \in \{0.1, 0.2, \dotsc, 1\},\; \alpha \in \{0.1, 0.2, \dotsc, 0.9\},\; \tfrac{m}{n} \in \{0.1, 0.2, \dotsc, 2\}\}$.

The encoder and decoder are implemented as a two-layer ReLU neural network with 128 hidden neurons and a latent dimension of 3, trained using squared loss. In the diffusion model, $\theta_u$ is defined as a five-layer neural network with 512 hidden neurons in each layer. The timesteps are encoded into a 64-dimensional embedding using a two-layer neural network. The diffusion process consists of 1000 steps, with noise levels ranging from $0.0001$ to $0.02$. The model is trained using the Adam optimizer for 5000 epochs with a learning rate of $0.0001$. For the regression task, we use a \texttt{RandomForest} regressor with 100 trees and default settings as implemented in \cite{scikit-learn}. For the SMOGN method, we use the validation set to select the best perturbation size in $\{0.02,0.04,0.06\}$.

\noindent \textbf{Sentiment analysis.} 
To generate review contexts, we condition the diffusion model on the review rating. Specifically, we first apply SentenceBERT to encode the original text into text embeddings. These embeddings, together with the rating vectors, are used to train a conditional diffusion model.

The timesteps are encoded into a time embedding using a two-layer fully connected neural network with a hidden dimension of 2048. At the core of the reverse diffusion process is a stack of ten residual blocks. Each block performs a linear transformation, followed by layer normalization and a ReLU activation. After iterative refinement through these blocks, a final projection layer maps the output back to the text embedding space. This enables the model to predict the denoised text embedding at each step of the diffusion process. The denoised embeddings are then decoded into natural language using the \texttt{vec2text} model (\url{https://huggingface.co/ielabgroup/vec2text_gtr-base-st_corrector}) to generate the final review texts.

The diffusion model is trained for 1000 epochs using the Adam optimizer with a fixed learning rate of $1 \times 10^{-4}$ and a batch size of 256.

For the sentiment classification task, we employ a multi-class logistic regression model as the classifier. Text reviews are transformed into numerical features using a TF-IDF vectorizer.

\noindent \textbf{Age prediction.}
We use Stable Diffusion 2.0 as the generator to generate face images from prompts. Specifically, the prompt is set as "A \{age\}-year-old \{gender\} of \{ethnicity\}", where "age" is the age sampled from the age range from the data uniformly, the gender is male or female, and the ethnicity belongs to the set $\{\text{White},\text{Black},\text{Asian},\text{Indian},\text{other}\}$. To fine-tune the diffusion model, we adopt the PEFT-LoRA \citep{peft} to fine-tune the U-net in the stable diffusion for 20000 steps. The age regressor is set as a convolutional neural network designed for image inputs of size \(64 \times 64\) with 3 channels (e.g., RGB images). It consists of:
Four convolutional layers with ReLU activations and max pooling;
Two fully connected layers with ReLU and dropout; A final output layer for regression.
The network progressively reduces the spatial resolution through pooling while increasing feature depth, eventually flattening into a fully connected regressor. The output is a scalar, making the model suitable for continuous-valued predictions.

\begin{table}[h!]
\centering
\caption{Summary of CNN Architecture (Input: \(3 \times 64 \times 64\))}
\begin{tabular}{llc}
\hline
\textbf{Layer Type} & \textbf{Details} & \textbf{Output Shape} \\
\hline
Input & 3-channel image & \(3 \times 64 \times 64\) \\
\hline
Conv2d + ReLU + MaxPool & \(3 \rightarrow 32\), kernel \(3\times3\), pool \(2\times2\) & \(32 \times 32 \times 32\) \\
\hline
Conv2d + ReLU + MaxPool & \(32 \rightarrow 64\), kernel \(3\times3\), pool \(2\times2\) & \(64 \times 16 \times 16\) \\
\hline
Conv2d + ReLU + MaxPool & \(64 \rightarrow 128\), kernel \(3\times3\), pool \(2\times2\) & \(128 \times 8 \times 8\) \\
\hline
Conv2d + ReLU + MaxPool & \(128 \rightarrow 256\), kernel \(3\times3\), pool \(2\times2\) & \(256 \times 4 \times 4\) \\
\hline
Flatten & — & \(4096\) \\
\hline
Linear + ReLU & \(4096 \rightarrow 256\) & 256 \\
\hline
Dropout & \(p = 0.4\) & 256 \\
\hline
Linear + ReLU & \(256 \rightarrow 256\) & 256 \\
\hline
Output (Linear) & \(256 \rightarrow 1\) & 1 \\
\hline
\end{tabular}
\end{table}

\noindent \textbf{Anomaly detection.} The autoencoder is a fully connected neural network designed for image compression and reconstruction on grayscale inputs of size $28 \times 28$, such as MNIST digits. It consists of two main components: an encoder and a decoder.

The encoder first flattens the input image into a 784-dimensional vector. This is followed by two fully connected layers: the first maps the input to a 128-dimensional hidden representation with ReLU activation, and the second compresses it to a 32-dimensional latent vector, also using ReLU.

The decoder reconstructs the image from this 32-dimensional latent code. It begins with a fully connected layer that maps the latent space back to 128 dimensions with ReLU activation, and then a second linear layer that projects it to 784 dimensions. A sigmoid activation is applied at the final layer to ensure that the reconstructed pixel values lie in the range \([0, 1]\). The output can be reshaped externally back to the original image size of \(28 \times 28\).

In fine-tuning the diffusion model to generate USPS data, we use the Adam optimizer with a learning rate of 0.0001 for 20 epochs.

\section{Additional simulations}

\subsection{Sample split.} 
\label{sample-split}
\noindent \textbf{Cross-fit v.s. single-fit.} A natural question is whether cross-fitting can enhance Algorithm 1 (CoDSA). For a fixed split ratio $r=0.8$, we test a five-fold scheme: partition the real data into \(K\) folds, train the generator on \(K\!-\!1\) folds, generate synthetics, fit the estimator on the held-out fold plus synthetics, and average the \(K\) estimates. 

In the regression non-transfer setting of Section 4.2 with split ratio \(r=0.8\), cross-fitting lowered variance and improved RMSE by $1.53\%$ percentage point,as suggested by Table~\ref{tab:cross} and Figure~\ref{fig:cross} below. However, $K$-fold cross-fitting multiplies runtime and peak GPU memory requirements and introduces cross-fold dependencies among the synthetic data, complicating risk analysis; thus, we retain the simpler single-split strategy in the main text.

\begin{table}[h]
\centering
\caption{RMSE of cross-fit versus single-split across hyperparameter values at split ratio $r=0.8$; all 
other parameters tuned by cross-validation.}
\label{tab:cross}
\begin{tabular}{lcc}
\hline
        & \textbf{Cross-fit} & \textbf{Single-fit} \\
\hline
RMSE    & 1.0560 (0.2320)      & 1.0724 (0.3452)       \\
\hline
\end{tabular}
\end{table}

\begin{figure}[ht]
\centering
\includegraphics[width=0.9\textwidth]{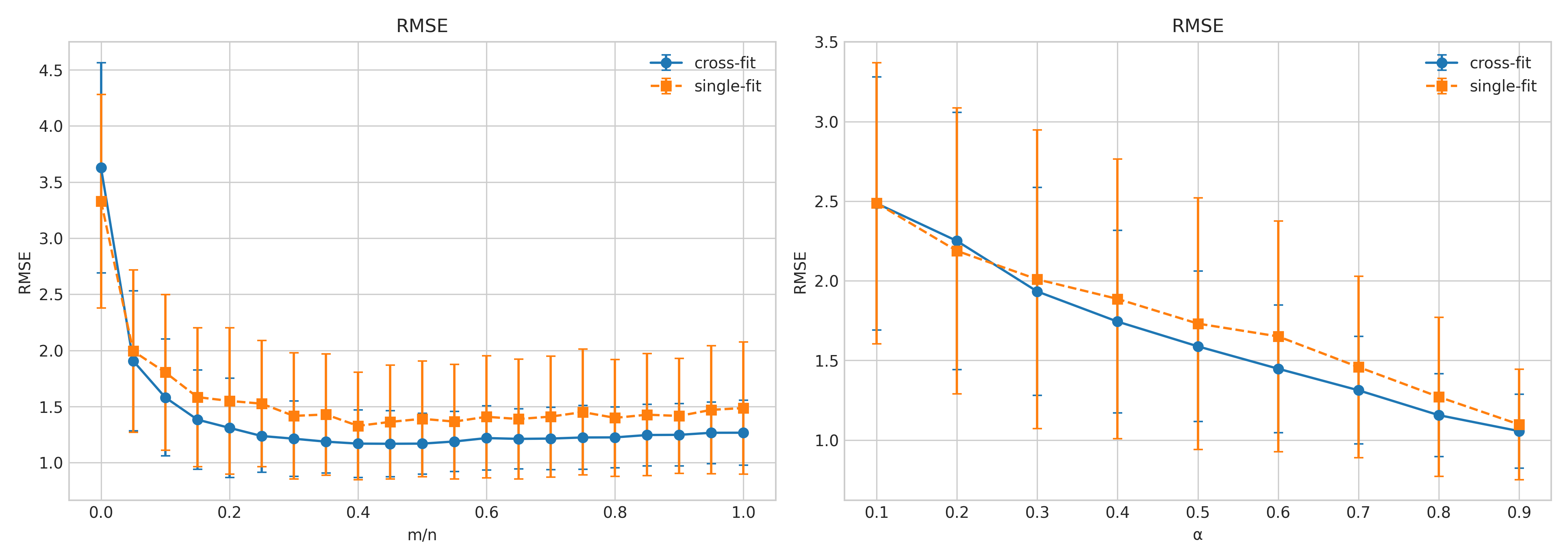}
\caption{Comparison of RMSE between cross-fit and single-fit strategies as a function of $m/n$ and $\alpha_1$ while $r=0.8$.\label{fig:cross}}
\end{figure}

\subsection{Sensitivity to pretraining sample size.}
\noindent \textbf{The Impact of Pretraining Data Size.}
In this experiment, we assess CoDSA’s performance on the regression task described in Section 4.2 as a function of the pretraining dataset size. We vary the source data size over \{1000, 3000, 6000, 10000\} samples. Figure~\ref{fig:sen} presents boxplots of the test RMSE across 10 replicates. As the number of pretraining samples increases, both the stability and accuracy of CoDSA improve. This confirms that larger pretraining datasets from a similar source task enable the construction of higher-fidelity generative models, which then enhance downstream CoDSA performance.

\begin{figure}[ht]
\centering
\includegraphics[width=0.5\textwidth]{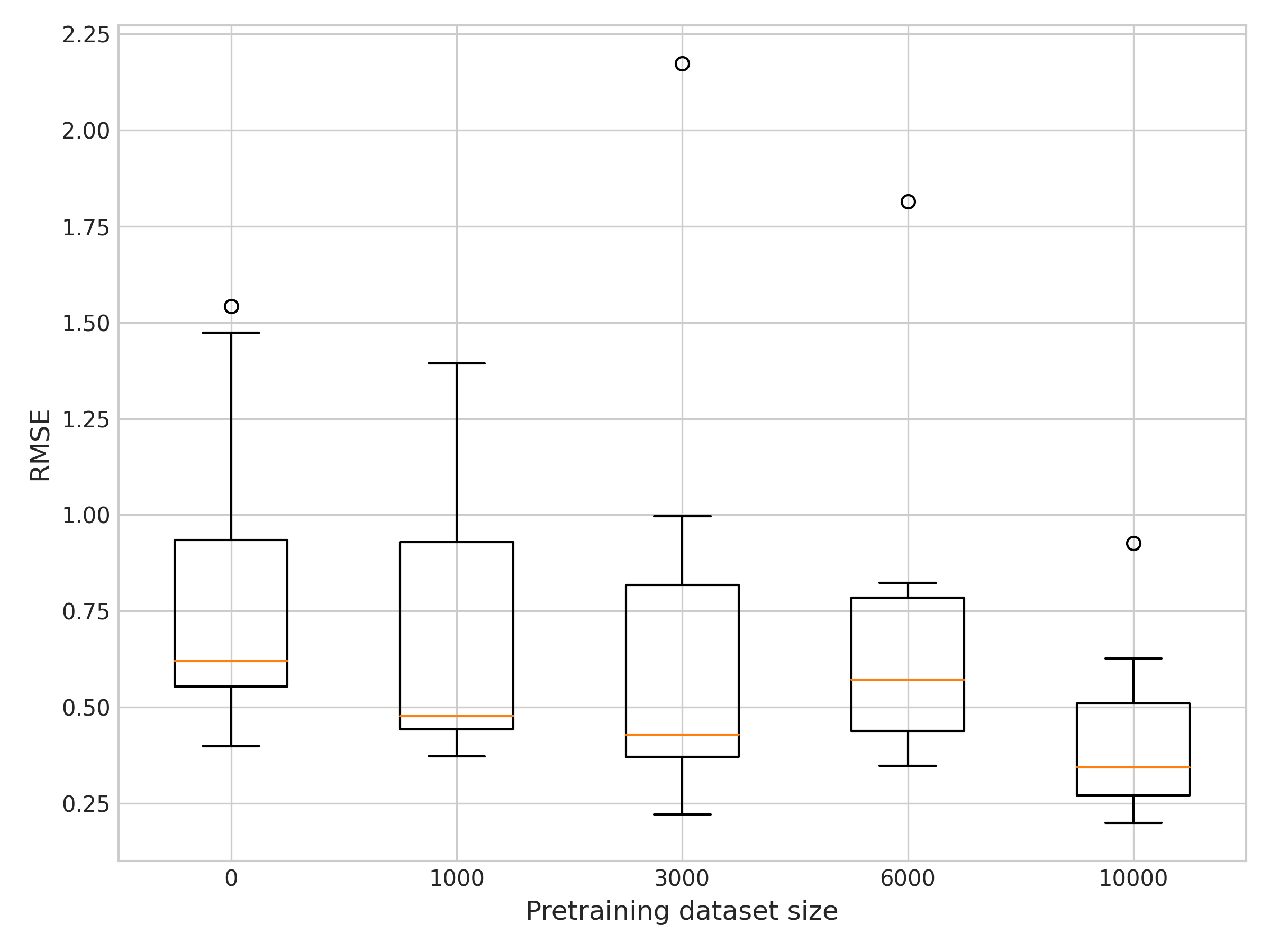}
\caption{Test RMSE under varying pretraining dataset sizes.}
\label{fig:sen}
\end{figure}

\bibliographystyle{Chicago}

\end{document}